\pdfoutput=1

\documentclass[letterpaper,10pt,onecolumn]{article} 

\usepackage{amsthm}
\usepackage{booktabs}
\usepackage[numbers]{natbib}
\usepackage{subcaption}

\usepackage[ruled,norelsize]{algorithm2e}

\SetAlFnt{\small}
\SetAlCapFnt{\small}
\SetAlCapNameFnt{\small}
\SetAlCapHSkip{0pt}
\IncMargin{-\parindent}

\usepackage{color}
\usepackage{latexsym}
\usepackage{listings}
\usepackage{enumitem}
\usepackage{amsfonts,amsmath,amssymb}
\usepackage{tabularx}
\usepackage{cprotect}
\usepackage{multirow}
\usepackage{adjustbox}
\usepackage{microtype} 
\usepackage[utf8]{inputenx}
\usepackage{thmtools,thm-restate}
\usepackage{seqsplit}
\usepackage[breakable,most]{tcolorbox}
\usepackage{longtable}
\usepackage{hyperref}
\usepackage{cleveref} 

\newcommand{\oldrevision}[1]{}

\newtheorem{definition}{Definition}
\newtheorem{example}{Example}

\let\oldnl\nl 
\newcommand{\nonl}{\renewcommand{\nl}{\let\nl\oldnl}} 

\newcommand{\mytt}[1]{\text{${\tt #1}$}}

\newcommand*\rot{\rotatebox{90}}
\newcolumntype{C}[1]{>{\centering\arraybackslash}m{#1}}
\newcolumntype{R}[1]{>{\raggedleft\arraybackslash}m{#1}}

\newcommand{\mychapter}{Chapter~}
\newcommand{\mysection}{Section~}
\newcommand{\mydefinition}{Definition~}
\newcommand{\mydefinitions}{Definitions~} 
\newcommand{\myfigure}{Figure~}
\newcommand{\mytable}{Table~}
\newcommand{\mylemma}{Lemma~}
\newcommand{\mylemmas}{Lemmas~} 
\newcommand{\mytheorem}{Theorem~}
\newcommand{\mycorollary}{Corollary~}
\newcommand{\myequation}{Equation~}
\newcommand{\myalgorithm}{Algorithm~}

\newcommand{\trans}{\Rightarrow}
\newcommand{\mysem}[1]{[\![#1]\!]}

\newcommand{\live}{\mytt{live}}
\newcommand{\islive}{\mytt{is\_live}}
\newcommand{\ureachdef}{\mytt{urdef}}

\newcommand{\stmt}{\mytt{stmt}}
\newcommand{\point}{\mytt{point}}
\newcommand\Tau{\mathcal{T}}
\newcommand{\wskip}{\mytt{skip}}

\newcommand{\wdef}{\mytt{def}}
\newcommand{\wuse}{\mytt{use}}
\newcommand{\wtrans}{\mytt{trans}}
\newcommand{\wfreevar}{\mytt{freevar}}
\newcommand{\wconlit}{\mytt{conlit}}
\newcommand{\wstmt}{\mytt{stmt}}
\newcommand{\wpoint}{\mytt{point}}
\newcommand{\wx}{\mytt{x}}
\newcommand{\wy}{\mytt{y}}
\newcommand{\we}{\mytt{e}}
\newcommand{\wc}{\mytt{c}}
\newcommand{\wv}{\mytt{v}}
\newcommand{\wm}{\mytt{m}}

\newcommand{\mundef}{unde\hspace{-0.1em}f}
\newcommand{\osrtrans}{\mytt{OSR\_trans}}
\newcommand{\buildcomp}{\mytt{build\_comp}}
\newcommand{\reconstruct}{\mytt{reconstruct}}
\newcommand{\apply}{\mytt{apply}}
\newcommand{\dopasses}{\mytt{do\_passes}}

\newcommand{\osrkit}{{\sf OSRKit}}
\newcommand{\tinyvm}{{\sf TinyVM}}
\newcommand{\clang}{{\tt clang}}

\newcommand{\alloca}{\mytt{alloca}}
\newcommand{\load}{\mytt{load}}
\newcommand{\memtoreg}{\mytt{mem2reg}}
\newcommand{\store}{\mytt{store}}
\newcommand{\RAUW}{\mytt{replaceAll}}
\newcommand{\RAUWfull}{\mytt{replaceAll}$(O, N)$}
\newcommand{\gdb}{\mytt{gdb}}
\newcommand{\lldb}{\mytt{LLDB}}

\newcommand{\fbase}{$\textsf{f}$}
\newcommand{\fvariant}{$\textsf{f'}$}
\newcommand{\fosrfrom}{$\textsf{f}_{\textsf{from}}$}
\newcommand{\fosrto}{$\textsf{f}\textsf{'}_{\textsf{to}}$}

\newcommand{\osrsource}{$\textsf{L}$}
\newcommand{\osrlanding}{$\textsf{L'}$}

\newcommand{\speccpu}{{\tt SPEC}~{\tt CPU2006}}
\newcommand{\phoronixpts}{{\tt Phoronix}~{\tt PTS}}

\begin{document}


\title{On-Stack Replacement \`{a} la Carte\vspace{1em}}
\date{}
\author{
  Daniele Cono D'Elia\\
  \texttt{delia@dis.uniroma1.it}
  \and
  Camil Demetrescu\\
  \texttt{demetres@dis.uniroma1.it}
}

\maketitle

\begin{abstract}

On-stack replacement (OSR) dynamically transfers execution between different code versions. This mechanism is used in mainstream runtime systems to support adaptive and speculative optimizations by running code tailored to provide the best expected performance for the actual workload. Current approaches either restrict the program points where OSR can be fired or require complex optimization-specific operations to realign the program's state during a transition. 
The engineering effort to implement OSR and the lack of abstractions make it rarely accessible to the research community, leaving fundamental question regarding its flexibility largely unexplored.


In this article we make a first step towards a provably sound abstract framework for OSR. We show that compiler optimizations can be made OSR-aware in isolation, and then safely composed. We identify a class of transformations, which we call {\em live-variable equivalent} (LVE), that captures a natural property of fundamental compiler optimizations, and devise an algorithm to automatically generate the OSR machinery required for an LVE transition at arbitrary program locations.

We present an implementation of our ideas in LLVM and evaluate it against prominent benchmarks, showing that bidirectional OSR transitions are possible almost everywhere in the code in the presence of common, unhindered global optimizations. We then discuss the end-to-end utility of our techniques in source-level debugging of optimized code, showing how our algorithms can provide novel building blocks for debuggers for both executables and managed runtimes.

\end{abstract}


\section{Introduction}
\label{se:introduction}

On-stack replacement (OSR) is a mechanism employed in language runtimes to dynamically switch the execution between different versions of a function~\cite{Paleczny01,Fink03}. Modern runtimes typically generate multiple variants of a function with different, often speculative, optimizations, adapting the code to execute to the current workload~\cite{Alpern00,Paleczny01}. OSR is usually at the core of large and complex just-in-time (JIT) compilers employed by popular production virtual machines (VMs), and is essential technology for dynamic optimization and debugging. Due to the substantial engineering effort to implement it, OSR tends to be restricted to a few of the most advanced runtime systems and is rarely accessible to the research community. The situation is further complicated by the lack of language abstractions to reason about the correctness and the flexibility of the mechanism. Common OSR embodiments require VM designers to manually generate ad-hoc metadata and glue code to get the program state to a correct resumption point. 
Other implementations restrict OSR transitions to places where hopefully there is no need to fix the program's state, e.g., at a function's entry point or at a loop's header. This results in lack of flexibility, making current approaches hardly applicable to several scenarios such as the ones we discuss below.

\paragraph{Motivating Examples} Our first example is excerpted from the Java HotSpot Glossary of Terms~\cite{hotspot-glossary}: ``A compiler initially assumes a reference value is never null, and tests for it using a trapping memory access. Later on, the application uses null values, and the method is deoptimized and recompiled to use an explicit test-and-branch idiom to detect such nulls''. As a null reference exception may be thrown anywhere in the code, deoptimization cannot happen on the fly unless OSR can be performed at arbitrary program locations. As a second example, we wish to collect accurate information about a program crash in an optimized production environment. When a crash happens, OSR reverts the program's execution to the state it would have had in the original unoptimized version and creates an informative core dump that includes the values of live variables that appear in the source code at that point. For this to work, OSR should be made to work at arbitrary locations in arbitrarily optimized programs. In a third scenario, we obfuscate a program to prevent security attacks by randomly diverting execution between different versions of a program at arbitrary execution points.

\medskip\noindent This article investigates how to overcome the limitations of previous approaches in the literature, supporting OSR at arbitrary program points across multiple unhindered program transformations.




\paragraph{Contributions and Overview} We contribute to the theory and practice of OSR by addressing a number of fundamental questions regarding its underlying computation model and how it can be mapped to concrete efficient implementations.
We provide the first formal treatment of OSR, distilling its essence to an abstract program morphing problem over a simple imperative calculus with an operational semantics. Our formalization aims at bridging the gap between the engineering practice in which OSR has been incubated, and the formal language methods, providing tools for reasoning abstractly and devising provably sound techniques.

To capture OSR in its full generality, we define a notion of {\em multi-program}, i.e., a collection of different program versions along with support to dynamically transfer execution between them. Using program bisimulation, we show that an OSR can correctly divert execution from one version to the other if they are {\em live-variable bisimilar}, i.e., the live variables they have in common at  corresponding execution states are equal. We identify a class of {\em live-variable equivalent} transformations that captures a natural property of common fundamental compiler optimizations, and devise algorithms for them that can automatically generate machinery to support OSR transitions at arbitrary program points in constant time and space.

A {\em compensation code} fixes the program state so that execution can correctly resume after an OSR transition, reconstructing the values of the variables that are live at the OSR target but not at the source. We make single transformations OSR-aware in isolation, and flexibly combine them by exploiting the {\em composability} of compensation code. This has a direct practical impact, as it can provide VM builders with a rich ``menu'' of possible program points where OSR can safely occur, relieving them from the burden of manually generating compensation code. 


We present and evaluate an implementation of our ideas in LLVM, showing that our algorithms support OSR transitions almost everywhere in the code under several classic optimizations. We discuss the end-to-end utility of our techniques in source-level debugging of optimized code, providing novel building blocks for debuggers.
We show how to correctly report values expected at the source level for variables that have been optimized away or hold misleading information. This represents a step forward in the state of the art of optimized code debugging.
\paragraph{Structure of the Article} This article is organized as follows. In Section~\ref{se:language-framework} we define syntax and semantics of the simple imperative language we use to illustrate our ideas. We then present computation tree logic and rewrite rules to reason about program properties and describe code transformations. Section~\ref{se:framework} illustrates our theoretical framework for OSR: we devise algorithms for automatic compensation code generation, and propose a general OSR model based on the notion of multi-program. We discuss our LLVM implementation in Section~\ref{se:implementation}. Our case study on optimized code debugging is presented in Section~\ref{se:debugging}. Section~\ref{se:related} discusses the connections of our ideas with previous works. We consider directions for future work and present concluding remarks in Section~\ref{se:conclusions}. 


\section{Language Framework}
\label{se:language-framework}

Our discussion is based on a minimal imperative language whose syntax is reported in \myfigure\ref{fig:osr-program-syntax}. In this section we introduce some basic definitions used in our representation of programs, and provide a big-step semantics for the language. We then present a formalism based on {\em computation tree logic} (CTL) to reason about program properties and describe program transformations through rewrite rules with side conditions~\cite{Clarke86}. 

\subsection{Syntax}
\label{ss:syntax}

\begin{definition}[Program]
\label{de:program}
A program is a sequence of instructions of the form:
\vspace{-2mm} 
\begin{equation*}
\pi=\langle I_1, I_2, \ldots, I_n \rangle\in Prog = \bigcup_{i=2}^{\infty} Instr^{i}
\vspace{-4mm}
\end{equation*}
where: 

\begin{itemize}[itemsep=2pt,parsep=0pt,topsep=2pt]
\item \texttt{$I_i\in Instr$} is the $i$-th instruction of the program, indexed by program point \texttt{$i\in[1,n]$}
\item \texttt{$I_1$ $=$ in $\cdots$} is the initial instruction, \texttt{$I_n$ $=$ out $\cdots$} is the final instruction
\item \texttt{$\forall i\in[2,n-1]:$ $I_i$ $\neq$ in $\cdots$ $\wedge$ $I_i$ $\neq$ out $\cdots$}
\end{itemize}
\end{definition}

\noindent Instruction \texttt{in} must appear at the beginning of a program and specifies the variables that must be defined prior to entering the program. Similarly, \texttt{out} occurs at the end and specifies the variables that are returned as output. 

By \texttt{e[x]} we indicate that \texttt{x} is a variable of the expression \texttt{e}\,$\in Expr$. We also denote by $vars($\texttt{e}$)$ the set of variables that occur in expression \texttt{e}. By $|\pi|=n$ we indicate the number of instructions in $\pi=\langle I_1, I_2, \ldots, I_n \rangle$.

\begin{figure}[!ht]
\vspace{0.25mm}
\begin{small}
\begin{center}
$\begin{array}{rcl}
\texttt{$Instr$} & ::= & \hphantom{\texttt{| }}\texttt{$Var$ := $Expr$} \\
& & \texttt{| if ( $Expr$ ) goto $Num$} \\
& & \texttt{| goto $Num$} \\
& & \texttt{| skip} \\
& & \texttt{| abort} \\
& & \texttt{| in $Var\cdots Var$} \\ 
& & \texttt{| out $Var\cdots Var$} \\
\texttt{ $Expr$ } & ::= & \texttt{$Num$ | $Var$ | $Expr$ + $Expr$ | $\ldots$ } \\
\texttt{ $Var$ } & ::= & \texttt{X | Y | Z | $\ldots$} \\
\texttt{ $Num$ } & ::= & \texttt{$\ldots$ | -2 | -1 | 0 | 1 | 2 | $\ldots$} \\
\end{array}$
\end{center}
\end{small}
\vspace{-3.5mm}
\caption{\label{fig:osr-program-syntax}Program Syntax}
\vspace{-2mm} 
\end{figure}

\subsection{Semantics}
\label{ss:semantics}

\begin{definition}[Memory Store]
A {\em memory store} is a total function $\sigma:Var\rightarrow \mathbb{Z}\cup\{\bot\}$ that associates integer values to defined variables, and $\bot$ to undefined variables. We denote by $\Sigma$ the set of all possible memory stores. 
\end{definition}

\noindent By $\sigma[\wx\gets v]$ we denote the same memory store function as $\sigma$, except that $\wx$ takes value $v$. Furthermore, for any $A\subseteq Var$, $\sigma\vert_{A}$ denotes $\sigma$ restricted to the variables in $A$, i.e., $\sigma\vert_{A}(\wx)=\sigma(\texttt{x})$ if $\wx\in A$ and $\sigma\vert_{A}(\wx)=\bot$ if $\wx\not\in A$. 

\begin{definition}[Program State]
\label{de:prog-state}
The {\em state} of a program $\pi=\langle I_1, I_2, \ldots, I_n \rangle$ is described by a pair $(\sigma,l)$, where $\sigma$ is a memory store and $l\in [1,n]$ is the program point of the next instruction to be executed. We denote by $State=\Sigma\times \mathbb{N}$ the set of all possible program states.
\end{definition}

\noindent We provide a big-step semantics using the transition relation $\trans_{\pi}\:\subseteq State\times State$, which specifies how a single instruction of a program $\pi$ affects its state. Our description relies on the relation $\Downarrow\subseteq(\Sigma\times Expr)\times \mathbb{Z}$ to describe how expressions are evaluated in a given memory store.

\begin{definition}[Big-Step Transitions]
\label{de:transitions}
For any program $\pi$, we define the relation $\Rightarrow_{\pi}\:\subseteq State\times State$ as follows, with meta-variables $\texttt{x}, \texttt{y}\in Var$, $\texttt{e}\in Expr$, and $\texttt{m}\in Num$:

\begin{small}
\begin{equation}
\label{eq:asgn-sem}
\frac
{I_l=\texttt{x:=e} ~~ \wedge ~~ (\sigma, \texttt{e}) \Downarrow v}
{(\sigma, l)\Rightarrow_{\pi} (\sigma[\wx\gets v], l+1)}
\end{equation}
\vspace{0.5mm}
\begin{equation}
\label{eq:goto-sem}
\frac
{I_l=\texttt{goto m}}
{(\sigma, l)\Rightarrow_{\pi} (\sigma, \texttt{m})}
\end{equation}
\vspace{0.5mm}
\begin{equation}
\label{eq:skip-sem}
\frac
{I_l=\texttt{skip}}
{(\sigma, l)\Rightarrow_{\pi} (\sigma, l+1)}
\end{equation}
\begin{equation}
\label{eq:ifz-sem}
\frac
{I_l=\texttt{if (e) goto m} ~~ \wedge ~~ (\sigma, \texttt{e}) \Downarrow 0}
{(\sigma, l)\Rightarrow_{\pi} (\sigma, l+1)}
\end{equation}
\vspace{0.5mm}
\begin{equation}
\label{eq:ifnz-sem}
\frac
{I_l=\texttt{if (e) goto m} ~~ \wedge ~~ (\sigma, \texttt{e}) \Downarrow v ~~~ \wedge ~~~ v\neq 0}
{(\sigma, l)\Rightarrow_{\pi} (\sigma, \texttt{m})}
\end{equation}
\vspace{0.5mm}
\begin{equation}
\label{eq:in-sem}
\frac
{I_1=\texttt{in x y}~\cdots ~~ \wedge ~~~ \sigma(\texttt{x})\neq\bot ~~~ \wedge ~~~ \sigma(\texttt{y})\neq\bot ~~~ \wedge ~~~ \cdots }
{(\sigma, 1)\Rightarrow_{\pi} (\sigma, 2)}
\end{equation}
\begin{equation}
\label{eq:out-sem}
\frac
{I_n=\texttt{out x y}~\cdots ~~ \wedge ~~~ \sigma(\texttt{x})\neq\bot ~~~ \wedge ~~~ \sigma(\texttt{y})\neq\bot ~~~ \wedge ~~~ \cdots }
{(\sigma, n)\Rightarrow_{\pi} (\sigma\vert_{\{\texttt{x}, \texttt{y}, \cdots\}}, n+1)}
\end{equation}
\end{small}
\vspace{0.5mm} 
\end{definition}

\noindent For a transition to apply, we assume that $I_l$ is defined, i.e., $l\in[1,n]$. 


\begin{definition}[Program Semantic Function]
\label{de:program-semantics}
We define the semantic function $\mysem{\pi}:\Sigma \rightarrow \Sigma$ of a program $\pi$ as: 
\begin{gather*}
\forall \sigma\in\Sigma: ~~ \mysem{\pi}(\sigma)=\sigma'~~
\Longleftrightarrow ~~ (\sigma,1) \Rightarrow^{*}_{\pi} (\sigma',|\pi|+1)
\end{gather*}
where $\Rightarrow^{*}_{\pi}$ is the transitive closure of $\Rightarrow_{\pi}$.
\end{definition}

\noindent Note that a program has undefined semantics if its execution on a given store does not reach the final \texttt{out} instruction. This accounts for infinite loops, abort instructions, exceptions, and ill-defined programs or input stores. 
We define the notion of program semantic equivalence as follows:

\begin{definition}[Program Equivalence]
\label{de:semantic-equivalence}
Two programs $\pi_1$ and $\pi_2$ are {\em semantically equivalent} iff $\mysem{\pi_1}=\mysem{\pi_2}$.
\end{definition}

\noindent A notion that will be useful in our framework is that of a {\em trace} of a transition system:

\begin{definition}[Traces]
\label{de:exec-trace}
A {\em trace} in a transition system $(S,$ $R\subseteq S^2)$ starting from $s\in S$ is a sequence $\tau=\langle s_0,s_1,\ldots,$ $s_i,\ldots\rangle$ such that $s_0=s$ and $\forall i\ge 0:~s_i\in\tau ~ \wedge ~ s_i~R~s_{i+1}$ $\Longleftrightarrow s_{i+1}\in\tau$. By ${\mathcal T}_{R,s}$ we denote the system of all traces of $(S,R\subseteq S^2)$ starting from $s$. By $\tau[i]$ we denote the $i$-th state of a trace $\tau$, i.e., $\tau[i]=s_i$. Furthermore, if $\tau$ is finite then $|\tau|$ denotes the index of its final state, i.e., $\tau=\langle s_0,s_1,\ldots,s_{|\tau|}\rangle$, otherwise $|\tau|=\infty$. Finally, $dom(\tau)=\{i: s_i\in\tau\}$ denotes the set of indexes of states in $\tau$.
\end{definition}

\noindent Notice that since $\Rightarrow_{\pi}$ is deterministic in our language, then for any initial store $\sigma$, the system of traces ${\mathcal T}_{\Rightarrow_{\pi},(\sigma,1)}$ of the execution transition system $(Store,\Rightarrow_{\pi})$ contains a single trace, which we denote by $\tau_{\pi\sigma}$.

\subsection{Reasoning about Program Properties}
\label{ss:reasoning}

To analyze properties of a program, we use Boolean formulas with free meta-variables that combine facts that must hold globally or at certain points of a program. Formulas can be checked against concrete programs by a {\em model checker}. For any program $\pi$ and formula $\phi$, the checker verifies whether there exists a substitution $\theta$ that binds free meta-variables with program objects so that $\theta(\phi)$ is satisfied in $\pi$. 
In this article, by $\mathcal{A}\models \phi$ we mean that $\phi$ is true in $\mathcal{A}$, i.e., formula $\phi$ is satisfied by structure $\mathcal{A}$ (or equivalently, $\mathcal{A}$ models $\phi$)~\cite{Clarke86}. 

Two global predicates that we will use later on are ${\tt conlit}(\wc)$, which states that an expression $\wc$ is a constant literal, and ${\tt freevar}(\wx,\we)$, which holds if and only if $\wx$ is a free variable of the expression $\we$.

\oldrevision{
To support analyses based on facts that involve finite maximal paths in the control flow graph, such as liveness and dominance, we use formulas based on CTL operators. In order to introduce these operators, we need to formalize the concept of finite maximal paths first.

\begin{definition}[Set of Complete Paths] Given a control flow graph $G=(V,E)$ and an initial node $n_0\in V$, the {\em set of complete paths} $CPaths(n_0,G)$ starting at $n_0$ consists of all finite sequences $\langle n_0,n_1,\ldots,n_k\rangle$ such that $(n_i,n_{i+1})\in E$ for all $n_i$ with $i<k$, and such that there does not exist a $n_{k+1}$ such that $(n_k,n_{k+1})\in E$.
\end{definition}

\noindent Complete paths from a specified node (i.e., instruction) are thus maximal finite sequences of connected nodes through a control flow graph from an initial point to a sink node, which in our setting is unique (unless {\tt abort} instructions are present) and corresponds to the final instruction $I_n$.
}

To support analyses based on facts that involve finite maximal paths in the control flow graph (CFG), such as liveness and dominance, we use formulas based on CTL operators.
First-order CTL can be used to specify properties of nodes and paths in a CFG. In particular, temporal CTL operators can be used to express properties of some or all possible future computational paths, any one of which might be an actual path that is realized.
We say that for any point $l$ in a program $\pi$ and two formulas $\phi$ and $\psi$, the following predicates are satisfied at $l$:
\begin{itemize}[parsep=0pt,topsep=3pt]
\item $\overrightarrow{AX}(\phi)$: if $\phi$ holds for all immediate successors of $l$;
\item $\overrightarrow{EX}(\phi)$: if $\phi$ holds for at least one immediate successor of $l$;
\item $\overrightarrow{A}(\phi~U~\psi)$: if $\phi$ holds on all paths from $l$, until $\psi$ holds;
\item $\overrightarrow{E}(\phi~U~\psi)$: if $\phi$ holds on at least one path from $l$, until $\psi$ holds.
\end{itemize}
\noindent Corresponding operators $\overleftarrow{AX}$ and $\overleftarrow{EX}$ are defined for immediate predecessors of $l$, while $\overleftarrow{A}$ and $\overleftarrow{E}$ refer to backward paths from $l$. Operators $A$ and $E$ are quantifiers over paths, while $X$ and $U$ path-specific quantifiers. Notice that $\phi~U~\psi$ requires that $\phi$ has to hold at least until at some node $\psi$ is satisfied: $\psi$ will thus be verified in the future.

\begin{figure}[!ht]
\begin{small}
\begin{center}
\begin{eqnarray*}
\wdef(\wx) & \triangleq & I_l= \texttt{x:=e} ~~ \vee ~~ I_l= \texttt{in} ~ \cdots ~ \texttt{x} \cdots \\
                            &            & [\wx ~ \textit{is defined by instruction} ~ I_l ~ \textit{in} ~ \pi] \\
\wuse(\wx) & \triangleq & I_l= \texttt{y:=e[x]} ~ \vee  \\
                            &            & I_l= \texttt{if (e[x]) goto m} ~ \vee \ \\
                            &            & I_l= \texttt{out} ~ \cdots ~ \texttt{x} \cdots \\
                            &            & [\wx ~ \textit{is used by instruction} ~ I_l ~ \textit{in} ~ \pi] \\
\stmt(I) & \triangleq & I=I_l ~~~ [I ~ \textit{is the instruction at} ~ l ~ \textit{in} ~ \pi]\\
\point(\texttt{m}) & \triangleq & \texttt{m}=l ~~~ [\textit{program point} ~ \wm ~ \textit{is} ~ l ~ \textit{in} ~ \pi] \\
\wtrans(\we) & \triangleq & I_l= \texttt{x:=e'} ~ \wedge ~ \neg\wfreevar(\wx,\we) ~ \vee I_l\neq\texttt{x:=e'}\\
                            &		& [\textit{no constituent of}~\we~\textit{is modified by instruction}~I_l ~ \textit{in} ~ \pi] \\                            
\islive(\wx) & \triangleq & \overleftarrow{AX}\overleftarrow{A}(\text{true} ~ U ~ \wdef(\wx))  \wedge\overrightarrow{E}(\neg\wdef(\wx) ~ U ~ \wuse(\wx)) \\
                            &            & [\wx ~ \textit{is live at program point} ~ l ~ \textit{in} ~ \pi] \\                            
\ureachdef(\wx,l') & \triangleq & \overleftarrow{AX}\overleftarrow{A}(\neg\wdef(\wx) ~ U ~ \point(l')\wedge\wdef(\wx)) \\
                            &            & [\textit{unique definition of}~\wx~{at}~l'~\textit{reaching}~l~\textit{in} ~ \pi] \\
\end{eqnarray*}
\end{center}
\end{small}
\vspace{-4mm}
\caption{\label{fig:osr-loc-pred}Predicates expressing local properties of a point $l\in [1,n]$ in a program $\pi=\langle I_1,\ldots,I_n\rangle$, with meta-variables $\texttt{e},\texttt{e'}\in Expr$, $\texttt{x}, \texttt{y}\in Var$, and $l, \texttt{m}\in Num$.}
\end{figure}


\myfigure\ref{fig:osr-loc-pred} shows a number of local predicates that will be useful throughout this article.
For instance, $\pi,l\models \ureachdef(\wx, l')$ holds if there is a {\em unique reaching definition} of $\wx$ that reaches $l$, and this definition is at $l'$. Its formulation states that on all backward paths ($\overleftarrow{A}$) starting at all the predecessors of $l$ ($\overleftarrow{AX}$), there is no node assigning to $x$ until $l'$ is reached. The following definition will be useful, too: 

\begin{definition}[Live Variables]
\label{de:live-var}
The set of live variables of a program $\pi$ at point $l$ is defined as:
\vspace{-1mm}
\begin{equation*}
\live(\pi,l) \triangleq \{ ~ \wx\in Var: \pi, l\models \islive(\wx) ~ \}
\end{equation*}
\end{definition}





\subsection{Program Transformations}
\label{ss:transformations}

To describe program transformations, we use rewrite rules with side conditions in a similar manner to~\cite{Lacey04,Kundu09}. We consider generalized rules that transform multiple instructions simultaneously, with side conditions drawn from CTL formulas:

\begin{definition}[Rewrite Rule]
\label{de:rewrite-rule}
A rule $T$ has the form:
\vspace{-1mm}
\begin{equation*}
\begin{array}{lllll}
T = & m_1: \hat{I}_1 \Longrightarrow \hat{I'}_1 
& \cdots
& m_r: \hat{I}_r \Longrightarrow \hat{I'}_r
& {\tt if} ~ \phi
\end{array}
\vspace{-1mm}
\end{equation*}
\noindent where $\forall k\in[1,r]$, $m_k$ is a meta-variable that denotes a program point, $\hat{I}_k$ and $\hat{I'}_k$ are program instructions that can contain meta-variables, and $\phi$ is a side condition that states whether the rewriting rule can be applied to the input program. We denote by $\Tau$ the set of all possible rewrite rules.
\end{definition}

\oldrevision{
\noindent An elementary example of rewrite rule with meta-variables $\wm$, $\wx$, and $\wy$ is: $$m: ~ {\tt y:=2*x} ~~ \Longrightarrow ~~ {\tt y:=x+x} ~~~ {\tt if} ~ true$$ which encodes a peephole optimization based on a weak form of operator strength reduction~\cite{Cooper01}.
}

\noindent Rules can be applied to concrete programs by a transformation engine based on model checking: when the checker finds a substitution $\theta$ that binds free meta-variables with program objects so that $\theta(\phi)$ is satisfied in $\pi$ and $\theta(\hat{I}_k)=I_{\theta(m_k)}\in \pi$ for some $k\in[1,t]$, then $I_{\theta(m_k)}$ is replaced with $\theta(\hat{I'}_k)=I'_{\theta(m_k)}\in \pi'$, as formalized next:

\begin{definition}[Rule Semantics]
\label{de:trans-func}
Let $T$ be a rewrite rule as in \mydefinition\ref{de:rewrite-rule}. The transformation function $\mysem{T}: Prog\rightarrow Prog$ is defined as follows:
\vspace{-2mm}
\begin{multline*}
\forall \pi, \pi'\in Prog: \pi'=\mysem{T}(\pi) \Longleftrightarrow
\exists ~ \theta: ~ \pi\models \theta(\phi) ~ \wedge ~ \\
\forall k\in[1,r]: \theta(\hat{I}_k)=I_{\theta(m_k)}\in \pi ~ \wedge ~ \theta(\hat{I'}_k)=I'_{\theta(m_k)}\in \pi'
\end{multline*}
We say that $T$ is {\em semantics-preserving} if for any program $\pi$ it holds $\mysem{\pi}=\mysem{\pi'}$, where $\pi'=\mysem{T}(\pi)$.
\end{definition}

\noindent In this article, we focus on transformations that do not alter the semantics of a program. Examples of semantics-preserving rules for classic compiler optimizations as proved in~\cite{Lacey02,Lacey04} are given in \myfigure\ref{fig:sample-trans}.

\begin{figure}[!hb]
\begin{center}
\begin{small}
\begin{minipage}[t]{0.47\textwidth}
\begin{tabularx}{0.9\textwidth}{X}{
\begin{tabularx}{\textwidth}{|X|}\hline
{\bf Constant propagation} (CP)\\\hline 
$m: ~ {\tt x:=e[v]} ~~ \Longrightarrow ~~ {\tt x:=e[c]}$ \\
${\tt if} ~~ \wconlit(\wc) ~ \wedge ~ m \models \overleftarrow{A}(\neg\wdef(\wv) ~ U ~ \wstmt({\tt v:=c}))$ \\\hline
\end{tabularx}
\newline
\newline
\begin{tabularx}{\textwidth}{|X|}\hline
{\bf Dead code elimination} (DCE)\\\hline 
$m: ~ {\tt x:=e} ~~ \Longrightarrow ~~ {\tt skip}$ \\
${\tt if} ~~ m \models \overrightarrow{AX} ~ \neg\overrightarrow{E}(true ~ U ~ \wuse(\wx))$ \\\hline
\end{tabularx}
\vspace{-2mm}
}
\end{tabularx}
\end{minipage}
\quad
\begin{minipage}[t]{0.47\textwidth}
\begin{tabularx}{\textwidth}{|X|}\hline
{\bf Code hoisting} (Hoist)\\\hline 
$p: ~ {\tt skip} ~~ \Longrightarrow ~~ {\tt x:=e}$ \\
$q: ~ {\tt x:=e} ~~ \Longrightarrow ~~ {\tt skip}$ \\
${\tt if} ~~ p \models \overrightarrow{A}(\neg\wuse(\wx) ~ U ~ \wpoint(q)) ~~ \wedge$ \\
$\hphantom{\texttt{if}} ~~ q \models \overleftarrow{A}((\neg\wdef(\wx)\vee\wpoint(q))\wedge \wtrans(e) ~ U ~ \wpoint(p))$ \\\hline
\end{tabularx}
\end{minipage}
\vspace{-2mm} 
\end{small}
\end{center}
\caption{\label{fig:sample-trans} Rewriting rules for CP, DCE, and Hoist transformations.} 
\vspace{-3mm}
\end{figure}

\section{On-Stack Replacement Framework}
\label{se:framework}

OSR consists in dynamically transferring execution from a point $l$ in a program $\pi$ to a point $l'$ in a program $\pi'$ so that execution can transparently continue from $\pi'$ without altering the original intended semantics of $\pi$. To model this behavior, we assume there exists a function that maps each point $l$ in $\pi$ where OSR can safely be fired to the corresponding point $l'$ in $\pi'$ from which execution can continue.

The OSR practice often makes the conservative assumption that $\pi'$ can always continue from the very same memory store as $\pi$~\cite{DElia16}. However, this assumption may reduce the number of points where sound OSR transitions can be fired. To overcome this limitation and support more aggressive OSR transitions, our model includes a {\em store compensation code} $\chi$ to be executed during an OSR transition from point $l$ in $\pi$ to point $l'$ in $\pi'$. The goal of the compensation code is to fix the memory store of $\pi$ at $l$ so that execution can safely continue in $\pi'$ from $l'$ with the fixed store. Note that if no compensation is needed for an OSR transition, $\mysem{\chi}$ is simply the identity function. We formalize these concepts in the next sections.


\subsection{OSR Mappings}
\label{ss:osr-mapping}

The machinery required to perform OSR transitions between two programs can be modeled as an {\em OSR mapping}:

\begin{definition}[OSR Mapping]
\label{de:osr-mapping}
For any $\pi,\pi'\in Prog$, an {\em OSR mapping} from $\pi$ to $\pi'$ is a (possibly partial) function
$\mu_{\pi\pi'}:[1,|\pi|]\rightarrow [1,|\pi'|]\times Prog$ such that:
\begin{gather*}
\forall \sigma\in\Sigma, \forall s_i=(\sigma_i,l_i)\in\tau_{\pi\sigma}: l_i\in dom(\mu_{\pi\pi'}),~\\
\exists \sigma'\in\Sigma, \exists s_j=(\sigma_j,l_j)\in\tau_{\pi'\sigma'}:\\
\mu_{\pi,\pi'}(l_i)=(l_j,\chi)~\wedge~\mysem{\chi}(\sigma_i\vert_{\live(\pi,l_i)})=\sigma_j\vert_{\live(\pi',l_j)}
\end{gather*}
A mapping is {\em strict} if $\sigma'=\sigma$. We call the set of all possible mappings between any pair of programs $OSRMap$.
\end{definition}



\noindent Intuitively, an OSR mapping provides the information required to transfer execution from any realizable state of $\pi$, i.e., an execution state that is reachable from some initial store by $\pi$, to a realizable state of $\pi'$. This definition is rather general, as a non-strict mapping allows execution to be transferred to a program $\pi'$ that is {\em not} semantically equivalent to $\pi$. For instance, $\pi'$ may contain speculatively optimized code, or just some optimized fragments of $\pi$~\cite{Guo11,Bala00,Gal09}. In such scenarios, execution in $\pi'$ can typically be invalidated by performing an OSR transition back to $\pi$ or to some other recovery program. Notice that \mydefinition\ref{de:osr-mapping} uses a weak notion of store equality restricted to live variables. To simplify the discussion, we assume that the memory store is only defined on scalar variables (we address memory \mytt{load} and \mytt{store} instructions in \mysection\ref{ss:load-store}). Hence, the behavior of a program only depends on the content of its live variables:

\begin{restatable}{theorem}{onlylivecount}
\label{thm:only-live-count}
For any program $\pi\in Prog$, any $\sigma,\sigma'\in\Sigma$, and any $l,l'\in \mathbb{N}$, it holds: 
$$
(\sigma,l)\Rightarrow_{\pi}(\sigma',l') ~~ \Longleftrightarrow ~~ (\sigma\vert_{\live(\pi,l)},l)\Rightarrow_{\pi}(\sigma'\vert_{\live(\pi,l')},l')
$$
\end{restatable}

\noindent Notice that $dom(\mu_{\pi\pi'})\subseteq [1,|\pi|]$ is the set of all possible points in $\pi$ where OSR transitions to $\pi'$ can be fired. If $\mu_{\pi\pi'}$ is partial, then there are points in $\pi$ where OSR cannot be fired. In \mysection{\ref{ss:osr-mapping-algorithms}} we present an algorithm whose goal is to minimize the number of these points.


\subsection{Live-Variable Equivalent Transformations}
\label{ss:osr-lve}

In this section we discuss sufficient properties for a compiler transformation to be turned into a provably correct building block of an OSR-aware compilation toolchain. We first need to introduce some formal machinery based on bisimilarity of programs. 


\begin{definition}[Program Bisimulation]
\label{de:bisimulation}
A relation $R\subseteq State\times State$ is a bisimulation relation between two programs $\pi$ and $\pi'$ if for any input store $\sigma\in \Sigma$ it holds:
\begin{align*}
s\in\tau_{\pi\sigma} ~ & \wedge ~~ s'\in \tau_{\pi'\sigma} ~~ \wedge ~~ s~R~s' \Longrightarrow \\
1) & ~~ s\Rightarrow_{\pi} s_1 ~~~ \Longrightarrow ~~~ s'\Rightarrow_{\pi'} s'_1 ~~ \wedge  ~~ s_1~R~s'_1 \\
2) & ~~ s'\Rightarrow_{\pi'} s'_1 ~~~ \Longrightarrow ~~~ s\Rightarrow_{\pi} s_1 ~~ \wedge  ~~ s_1~R~s'_1
\end{align*}
\end{definition}

\noindent Our notion of bisimulation between programs $\pi$ and $\pi'$ requires that $R$ be a bisimulation between transition systems $(\tau_{\pi\sigma}, \Rightarrow_{\pi})$ and $(\tau_{\pi'\sigma}, \Rightarrow_{\pi'})$ for any store $\sigma\in \Sigma$. This implies that for any $\sigma$, $\tau_{\pi\sigma}$ is finite if and only if $\tau_{\pi'\sigma}$ is finite; also, if they are finite, then they have the same length.
This assumption can be made without loss of generality, as equal length of traces can be enforced by padding programs with \wskip\ statements.

\begin{definition}[Partial State Equivalence]
\label{de:state-equiv-relation}
For any function $A:\mathbb{N}\rightarrow 2^{Var}$, the {\em partial state equivalence} relation $R_A\subseteq State\times State$ is defined as:
\begin{equation*}
R_A\triangleq\{ (s, s')\in State\times State:  ~s=(\sigma,l) ~ \wedge ~ s'=(\sigma',l) ~ \wedge ~ \sigma\vert_{A(l)} = \sigma'\vert_{A(l)} \}.
\end{equation*}
\end{definition}

\noindent Relation $R_A$ is clearly reflexive, symmetric, and transitive.

\begin{definition}[Live-Variable Bisimilar Programs]
\label{de:lvb-programs}
$\pi$ and $\pi'$ are {\em live-variable bisimilar} (LVB) if $R_{A}$ is a bisimulation relation between them, where $A=l\mapsto\live(\pi,l)\cap \live(\pi',l)$ is the function that yields for each program point $l$ the set of variables that are live at $l$ in both $\pi$ and $\pi'$.
\end{definition}

\noindent We can now formally define the class of transformations we are interested in as follows:

\begin{definition}[Live-Variable Equivalent Transformation]
\label{de:lve-trans}
A program transformation $T$ is {\em live-variable equivalent} (LVE) if for any program $\pi$, $\pi$ and $\mysem{T}(\pi)$ are live-variable bisimilar.
\end{definition}

\noindent Live-variable equivalence is a natural property of fundamental compiler optimizations that insert, delete, or move instructions around. Constant propagation, dead code elimination, and code hoisting as defined in \myfigure\ref{fig:sample-trans} are examples of LVE transformations. 
\begin{restatable}{theorem}{lvetransexamples}
\label{th:lve-trans-examples}
Transformations CP, DCE, and Hoist of \myfigure\ref{fig:sample-trans} are live-variable equivalent.
\end{restatable}

\noindent
The argument for the proof follows the bisimulation relations used in~\cite{Lacey02} to prove the transformations correct. For CP, $R$ is simply the identity relation, while for DCE and Hoist it is piecewise-defined on the indices of the traces. Further optimizations not formally discussed here are evaluated in \mysection\ref{ss:evaluation}.




\subsection{OSR Mapping Generation Algorithm}
\label{ss:osr-mapping-algorithms}
We now discuss how to automatically enhance an existing LVE transformation so that, given a base program $\pi$, it produces not only a rewritten program $\pi'=\mysem{T}(\pi)$, but also a forward OSR mapping $\mu_{\pi\pi'}$ from $\pi$ to $\pi'$ and a backward OSR mapping $\mu_{\pi'\pi}$ from $\pi'$ to $\pi$. The produced compensation code runs in $O(1)$ time and supports bidirectional OSR between $\pi$ and $\pi'$, enabling both optimization and deoptimization.

The proposed algorithm, which we call \osrtrans, is shown in \myalgorithm\ref{alg:osr-trans} and relies on two subroutines: 1) \apply\ (defined in Theorem~\ref{th:osr-trans-correctness} and in Section~\ref{ss:tracking-opt}) builds a program $\pi'$ by applying $T$ on $\pi$ and two functions $\Delta:[1,|\pi|]\rightarrow [1,|\pi'|]$, $\Delta':[1,|\pi'|]\rightarrow [1,|\pi|]$ that map OSR program points between $\pi$ and $\pi'$; 2) \buildcomp\ (\myalgorithm\ref{alg:osr-build-comp}) constructs the {\em store compensation code} to be included in the mappings. If any of the live variables at the OSR destination cannot be guaranteed to be correctly assigned, no entry is created (lines~\ref{line-trans:def}, \ref{line-trans:defprime} in \myalgorithm\ref{alg:osr-trans}) and the point will not be eligible for OSR. In \mysection\ref{ss:evaluation} we analyze experimentally the fraction of points for which a compensation code can be created by \buildcomp\ in a variety of prominent benchmarks.

\oldrevision{
\begin{enumerate}[itemsep=2pt,parsep=0pt,topsep=2pt]
 \item a program $\pi'=\mysem{T}(\pi)$;
 \item an OSR mapping $\mu_{\pi\pi'}$ from $\pi$ to $\pi'$;
 \item an OSR mapping $\mu_{\pi'\pi}$ from $\pi'$ to $\pi$.
\end{enumerate}

\noindent Mappings $\mu_{\pi\pi'}$ and $\mu_{\pi'\pi}$ produced by the algorithm are based on compensation code that runs in $O(1)$ time and support bidirectional OSR between $\pi$ and $\pi'$, enabling invalidation and deoptimization. The algorithm, which we call \osrtrans, is shown in \myalgorithm\ref{alg:osr-trans} and relies on two subroutines: \apply\ and \buildcomp.

\osrtrans\ then constructs a forward mapping $\mu_{\pi\pi'}$ from $l$ in $\pi$ to $\Delta(l)$ in $\pi'$ (lines~\ref{line-trans:foreach}--\ref{line-trans:def}), and a backward mapping $\mu_{\pi'\pi}$ from $l'$ in $\pi'$ to $\Delta'(l')$ in $\pi$ (lines~\ref{line-trans:foreachprime}--\ref{line-trans:defprime}). 
}

\begin{figure}[ht]
\IncMargin{2em}
\begin{algorithm}[H]
\DontPrintSemicolon
\LinesNumbered
\SetAlgoNoLine
\SetAlgoNoEnd
\SetNlSkip{1em} 
\Indm\Indmm
$\mathbf{algorithm} \> \> \osrtrans$($\pi, T$)$\rightarrow$($\pi'$,$\mu_{\pi\pi'}$,$\mu_{\pi'\pi}$):\;
\everypar={\nl}
\Indp\Indpp
$(\pi',\Delta,\Delta')\gets \texttt{apply}(\pi,T)$\tcc*[r]{$\Delta,\Delta'$ map program points between $\pi,\pi'$} 
\ForEach{$l\in dom(\Delta)$}{\label{line-trans:foreach}
    \lIf{$(\chi\gets\buildcomp(\pi,l,\pi',\Delta(l)))\neq\mundef$}{
	$\mu_{\pi\pi'}(l)\gets(\Delta(l),\chi)$\
    }\label{line-trans:def}
}
\ForEach{$l'\in dom(\Delta')$}{\label{line-trans:foreachprime}
    \lIf{$(\chi\gets\buildcomp(\pi',l',\pi,\Delta'(l')))\neq\mundef$}{
	$\mu_{\pi'\pi}(l')\gets(\Delta'(l'),\chi)$
    }\label{line-trans:defprime}
}
\Return{$(\pi',\mu_{\pi\pi'},\mu_{\pi'\pi})$}\;
\Indm\Indmm
\DecMargin{0.5em}
\caption{\label{alg:osr-trans} \osrtrans\ algorithm for OSR mapping construction.}
\IncMargin{0.5em}
\end{algorithm}
\vspace{-2mm}
\end{figure}

\begin{figure}[ht]
\vspace{-1.5mm}
\IncMargin{2em}
\begin{algorithm}[H]
\DontPrintSemicolon
\LinesNumbered
\SetAlgoNoLine
\SetAlgoNoEnd
\SetNlSkip{1em} 
\Indm\Indmm
$\mathbf{algorithm} \> \> \buildcomp$($\pi$, $l$, $\pi'$, $l'$)$\rightarrow\chi$:\;
\everypar={\nl}
\Indp\Indpp
$\chi\gets \textbf{in}~x_1~x_2~\cdots~x_k\,:\,\forall i\in[1,k]:\pi,l \models \live(x_i)$\;\label{line-bc:firstline}
mark all program points of $\pi'$ as unvisited\;\label{line-bc:mark}
$\mathbf{try}$
\ForEach{$\wx:~\pi',l'\models\live(\wx)\wedge\pi,l\models\neg\live(\wx)$}{\label{line-bc:foreach}
    $\chi \gets \chi \cdot \reconstruct(\wx, \pi, l, \pi', l', l')$\;\label{line-bc:reconstruct}
}
$\mathbf{catch~return}~\mundef$\;\label{line-bc:undef}
$\chi\gets \chi\cdot\textbf{out}~x_1~x_2~\cdots~x_{k'}:\forall i\in[1,k']:\pi',l'\models \live(x_i)$\;\label{line-bc:updatechi}
\Return{$\chi$}\;
\Indm\Indmm
\IncMargin{1.5em}
\caption{\label{alg:osr-build-comp} \buildcomp\ algorithm for compensation code construction.}
\DecMargin{1.5em}
\end{algorithm}
\vspace{-2mm}
\end{figure}

\begin{figure}[ht]
\vspace{-1.5mm}
\IncMargin{2em}
\begin{algorithm}[H]
\DontPrintSemicolon
\LinesNumbered
\SetAlgoNoLine
\SetAlgoNoEnd
\SetNlSkip{1em} 
\Indm\Indmm
$\mathbf{procedure} \> \> \reconstruct$($\wx$, $\pi$, $l$, $\pi'$, $l'$, $l''$):\;
\everypar={\nl}
\Indp\Indpp
\uIf{$\exists\hat{l}:\pi',l''\models\ureachdef(\wx,\hat{l})\wedge\pi',\hat{l}\models\stmt(\texttt{\rm x:=e})$}{\label{line-rec:first-inst}
    \lIf{$\hat{l}~\textsf{\rm is visited}$}{
	\hspace{-0.3em}\Return{$\langle\rangle$} 
    }
    mark $\hat{l}$ as visited\;
    \lIf{$\pi',l'\models\ureachdef(\wx,\hat{l})~\wedge~\pi',l'\models\live(x)~\wedge~\pi,l\models\live(x)$}{
    	\hspace{-0.3em}\Return{$\langle\rangle$}\label{line-rec:both-live} 
    }
    $\chi\gets \langle\rangle$\;
    \ForEach{$\wy:~\wy\in\wfreevar(\we)$}{\label{line-rec:freevar}
	$\chi \gets \chi \cdot \reconstruct(\wy, \pi, l, \pi', l', \hat{l})$\label{line-rec:recursive}
    }
    $\chi\gets \chi \cdot \texttt{\rm x:=e}$\;\label{line-rec:assign}
}
\lElse{
    \hspace{-0.2em}$\mathbf{throw}~\mundef$\label{line-rec:fail}
}
\Return{$\chi$}\;
\Indm\Indmm
\DecMargin{0.5em}
\caption{\label{alg:osr-reconstruct} Value reconstruction procedure used by \buildcomp.}
\IncMargin{0.5em}
\end{algorithm}
\vspace{-2mm}
\end{figure}

\oldrevision{
\paragraph{\texttt{OSR\_trans}} The algorithm relies on two subroutines: \apply\ and \buildcomp. Procedure \apply\ takes as input a program $\pi$ and a program rewriting function $T$, and returns a transformed program $\pi'$ and two functions $\Delta:[1,|\pi|]\rightarrow [1,|\pi'|]$, $\Delta':[1,|\pi'|]\rightarrow [1,|\pi|]$ that map OSR program points between $\pi$ and $\pi'$. Algorithm
\buildcomp\ (listed in \myalgorithm\ref{alg:osr-build-comp}) takes as input $\pi$, $l$, $\pi'$, $l'$ and aims to build a {\em store compensation code} $\chi$ that allows firing an OSR from $\pi$ at $l$ to $\pi'$ at $l'$. \osrtrans\ first calls \apply\ and then uses \buildcomp\ on $\pi$, $\pi'$, $\Delta$, $\Delta'$ to build the OSR mappings $\mu_{\pi\pi'},\mu_{\pi'\pi}$.

Lines~\ref{line-trans:foreach}--\ref{line-trans:def} build the forward mapping $\mu_{\pi\pi'}$ from $l$ in $\pi$ to $\Delta(l)$ in $\pi'$, while lines~\ref{line-trans:foreachprime}--\ref{line-trans:defprime} build the backward mapping $\mu_{\pi'\pi}$ from $l'$ in $\pi'$ to $\Delta'(l')$ in $\pi$. If any of the live variables at the OSR destination cannot be guaranteed to be correctly assigned, no entry is created in $\mu_{\pi\pi'}$ or $\mu_{\pi'\pi}$ for the OSR origin point (lines~\ref{line-trans:def} and~\ref{line-trans:defprime}). Hence, those points will not be eligible for OSR transitions. In \mysection\ref{se:evaluation} we analyze experimentally the fraction of points for which a compensation code can be created by \buildcomp\ in a variety of prominent benchmarks.
}

\paragraph{\texttt{build\_comp}} \myalgorithm\ref{alg:osr-build-comp} takes as input $\pi$, $\pi'$, and two locations $l$ and $l'$ to build a program $\chi$ that enables an OSR from $\pi$ at $l$ to $\pi'$ at $l'$. The ``{\tt in}'' statement spans the live variables at $l$ (line~\ref{line-bc:firstline}), while the ``{\tt out}'' statement yields the live variables at $l'$ (line~\ref{line-bc:updatechi}). The goal of $\chi$ is to make sure that all {\tt out} variables are correctly assigned, either because they already hold the correct value upon entry, or because they can be computed in terms of the input variables. The algorithm iterates over all the variables $x_i$ that are live at the destination, but not at the origin (line~\ref{line-bc:foreach}): procedure \reconstruct\ is called to build a code fragment that assigns $x_i$ with its correct value using live variables at the origin (line~\ref{line-bc:reconstruct}). On failure, an undefined compensation code is returned (line~\ref{line-bc:undef}), which implies that OSR cannot be performed at $l$. \reconstruct\ will mark points in $\pi'$ as visited to avoid duplicated code and unnecessary work. \buildcomp\ can be implemented with a running time linearly bounded by $|\pi'|$.

\paragraph{\texttt{reconstruct}} The procedure reported in \myalgorithm\ref{alg:osr-reconstruct} takes a variable \wx, the OSR origin and destination points $l$ and $l'$ in $\pi$ and $\pi'$, respectively, and an additional point $l''$ in $\pi'$. It builds a straight-line code fragment that assigns \wx\ with the value it would have had at $l''$ just before reaching $l'$ if execution had been carried on in $\pi'$ instead of $\pi$. The algorithm first checks whether there is a unique reaching definition of \wx\ of the form \mytt{x:=e} for point $l''$ at some point $\hat{l}$ in $\pi''$. In the presence of multiple reaching definitions, the algorithm gives up.
%
If \wx\ is live both at the origin $l$ and at the destination $l'$, and the definition of \wx\ at $\hat{l}$ that reaches $l''$ is also a unique reaching definition for $l'$ (line~\ref{line-rec:both-live}), then \wx\ would have assumed at $l''$ the same value available at $l'$. For the live-variable bisimilarity hypothesis, the algorithm correctly assumes that \wx\ is already available at the origin and no compensation code is needed to reconstruct it (\mytt{return} at line~\ref{line-rec:both-live}). If \wx\ is not available at $l$, then the algorithm iterates over all the constituents of the expression $e$ computed at $\hat{l}$ and recursively builds code that computes the values that they would have assumed at $\hat{l}$ just before reaching $l'$ if execution had been carried on in $\pi'$. Once the recursively generated code has been added to $\chi$, the assignment \mytt{x:=e} is appended to $\chi$ (line~\ref{line-rec:assign}).

\paragraph{Correctness}
Live-variable bisimilarity for $\pi$ and $\pi'$ is a sufficient condition for the correctness of \osrtrans:

\begin{restatable}{theorem}{osrtranscorrectness}
\label{th:osr-trans-correctness}
For any program $\pi$ and LVE transformation $T$, if ${\tt apply}(\pi,T)\triangleq(\pi', \Delta_I, \Delta_I)$ where $\pi'=\mysem{T}(\pi)$ and $\Delta_I:[1,|\pi|]\rightarrow [1,|\pi|]$ is the identity mapping between program points, then ${\tt OSR\_trans}(\pi,T)$ $=(\pi',\mu_{\pi\pi'},\mu_{\pi'\pi})$ yields a strict OSR mapping $\mu_{\pi\pi'}$ between $\pi$ and $\pi'$ and a strict OSR mapping $\mu_{\pi'\pi}$ between $\pi'$ and $\pi$.
\end{restatable}




\subsection{Composing Multiple Transformation Passes}
\label{ss:trans-compose}

A relevant property of OSR mappings is that they can be composed, allowing multiple optimization passes to be applied to a program using ${\tt OSR\_trans}$. The first ingredient is {\em program composition}, defined as follows:

\begin{definition}[Program composition]
\label{de:composition}
We say that two programs $\pi,\pi'\in Prog$ with $\pi=\langle I_1,\ldots,I_n\rangle$ and $\pi'=\langle I'_1,\ldots,I'_{n'}\rangle$ are {\em composable} if $I_n=\texttt{out}~v_1,\ldots,v_k$ and $I'_1=\texttt{in}~v'_1,\ldots,v'_{k'}$ with $\{v'_1,\ldots,v'_{k'}\}\subseteq\{v_1,\ldots,v_k\}$. For any pair of composable programs $\pi,\pi'$, we define $\pi\circ\pi'=\langle I_1,\ldots,I_{n-1},\hat{I'}_2,\ldots,\hat{I'}_{n'}\rangle$, where $\forall i\in[1,n']$, $\hat{I'}_i$ is obtained from $I'_i$ by relocating each {\tt goto} target $m$ with $m+n-2$.
\end{definition}

\oldrevision{
\begin{restatable}[Semantics of program composition]{lem}{progcompsem}
\label{le:prog-comp-sem}
Let $\pi,\pi'\in Prog$ be any pair of composable programs, then $\forall\sigma\in\Sigma,$ $\mysem{\pi\circ\pi'}(\sigma)=\mysem{\pi'}\left(\mysem{\pi}(\sigma)\right)$.
\end{restatable}
}

\noindent A composition of OSR mappings for composable programs can then be defined as follows:

\begin{restatable}[Mapping Composition]{theorem}{osrmappingcomp}
\label{thm:osr-mapping-comp}
Let $\pi,\pi',\pi''\in Prog$, let $\mu_{\pi\pi'}$ and $\mu_{\pi'\pi''}$ be OSR mappings as in \mydefinition\ref{de:osr-mapping}, and let $\mu_{\pi\pi'}\circ\mu_{\pi'\pi''}$ be a {\em composition of mappings} defined as follows:
\begin{gather*}
\forall l\in dom(\mu_{\pi\pi'}): \mu_{\pi\pi'}(l)=(l',\chi)\wedge l'\in dom(\mu_{\pi'\pi''}):\\
\mu_{\pi'\pi''}(l')=(l'',\chi')\implies(\mu_{\pi\pi'}\circ\mu_{\pi'\pi''})(l)=(l'',\chi\circ\chi')
\end{gather*}
Then $\mu_{\pi\pi'}\circ\mu_{\pi'\pi''}$ is an OSR mapping from $\pi$ to $\pi''$.
\end{restatable}

\oldrevision{
\begin{restatable}{cor}{composestrict}
\label{co:compose-strict}
Let $\pi,\pi',\pi''\in Prog$, let $\mu_{\pi\pi'}$ and $\mu_{\pi'\pi''}$ be strict OSR mappings as in \mydefinition\ref{de:osr-mapping}. Then $\mu_{\pi\pi'}\circ\mu_{\pi'\pi''}$ is a strict OSR mapping from $\pi$ to $\pi''$.
\end{restatable}
}

\subsection{Multi-Version Programs}
\label{se:multiprogram}

We conclude our formal treatment of OSR by proposing a general OSR model where computations are described by a {\em multi-version program}, which consists of different versions of a program along with OSR mappings to enable execution transfers between them. This captures possible OSR uses in their full generality.

\begin{definition}[Multi-Version Program]
\label{de:mv-program}
A multi-version program is an edge-labeled graph $\Pi=({\mathcal V}, {\mathcal E}, {\mathcal M})$ where ${\mathcal V}=\{ \pi_1, \pi_2, \ldots,\pi_r\}$ is a set of program versions, ${\mathcal E}\subseteq \Pi^2$ is a set of edges such that $(\pi_p,\pi_q)$ indicates that an OSR transition can be fired from some point of $\pi_p$ to $\pi_q$,
and ${\mathcal M}:{\mathcal E}\rightarrow OSRMap$ labels each edge $(\pi,\pi')\in {\mathcal E}$ with an OSR mapping from $\pi$ to $\pi'$.
\end{definition}


\noindent The state of a multi-version program is similar to the state of a program (\mydefinition\ref{de:prog-state}), but it also includes the index of the currently executed program version:


\begin{definition}[Multi-Version Program State]
The {\em state} of a multi-version program $\Pi=({\mathcal V}, {\mathcal E}, {\mathcal M})$ is described by a triple $(p,\sigma,l)$, where $p\in[1,|{\mathcal V}|]$ is the index of a program version, $\sigma$ is a memory store, and $l\in [1,|\pi_p|]$ is the point of the next instruction to be executed in $\pi_p$. The {\em initial state} from a store $\sigma$ is $(1,\sigma,1)$, i.e., computations start at $\pi_1$. We denote by $MState=\mathbb{N}\times\Sigma\times \mathbb{N}$ the set of all possible multi-version program states.
\end{definition}

\noindent A practical way to generate a multi-version program consists in starting from a base program and constructing a tree of different versions, where each version is derived from its parent by applying one or more transformations.

The execution semantics of a multi-version program is described by the following transition relation:

\begin{definition}[Multi-Version Big-Step Transitions]
\label{de:osr-semantics}
For any multi-version program $\Pi$, relation $\Rightarrow_{\Pi}\subseteq MState\times MState$ is defined as follows:

\begin{footnotesize}
\begin{equation}
\label{eq:mv-big-step}
\begin{array}{rc}
(Norm)
&
\dfrac
{(\sigma, l)\Rightarrow_{\pi_p} (\sigma',l')}
{(p,\sigma, l)\Rightarrow_{\Pi} (p,\sigma',l')}
\\
\\
(OSR)
&
\dfrac
{(\pi_p,\pi_q)\in {\mathcal E} ~ \wedge ~ (l',\chi)={\mathcal M}(\pi_p,\pi_q)(l) ~ \wedge ~ \sigma'=\mysem{\chi}(\sigma)}
{(p,\sigma, l)\Rightarrow_{\Pi} (q,\sigma',l')}\\
\end{array}
\end{equation}
\end{footnotesize}
\end{definition}

\noindent The meaning is that at any time, execution can either continue in the current program version (Norm rule), or an OSR transition -- if possible at the current point -- can direct the control to another program version (OSR rule). The choice is non-deterministic, i.e., an oracle can tell the execution engine which rule to apply.

In practice, the choice may be based for instance on profile data gathered by the runtime system: a common strategy is to dynamically ``OSR'' to the available version with the best expected performance on the actual workload. Notice that since $\Rightarrow_{\Pi}$ may be non-deterministic, in general there may be different final stores for the same initial store. However, we are interested here in multi-version programs that deterministically yield a unique result, which guarantees semantic transparency of OSR transitions:

\oldrevision{
Using this approach and the algorithm \dopasses\ described in \mysection\ref{ss:trans-compose}, it is straightforward to construct a multi-version program $\Pi=({\mathcal V}, {\mathcal E}, {\mathcal M})$ such that:
\vspace{-1mm}
\begin{align*}
(\pi_p,\pi_q)\in {\mathcal E} ~~ \Longleftrightarrow ~~ \exists L: ~ &{\tt do\_passes}(\pi_p,L)=(\pi_q,\mu,\mu') ~ \wedge ~ {\mathcal M}(\pi_p,\pi_q)=\mu ~~ \vee \\
&{\tt do\_passes}(\pi_q,L)=(\pi_p,\mu,\mu') ~ \wedge ~ {\mathcal M}(\pi_p,\pi_q)=\mu'
\end{align*}
}


\begin{restatable}[Multi-Version Program Determinism]{theorem}{mvprogdeterm}
\label{th:mv-prog-determ}
Let $\Pi$ $=({\mathcal V}, {\mathcal E}, {\mathcal M})$ be a multi-version program constructed using OSR mapping composition over LVE transformations. Then $\Pi$ is deterministic.
\end{restatable}


\subsection{Discussion}
\label{ss:osr-discussion}

\mytheorem\ref{thm:osr-mapping-comp} allows us to flexibly combine transformation rules, provided that an OSR mapping between the original and modified programs can be produced for each rule. \buildcomp\ can automatically generate compensation code required for LVE transformations, but the applicability of mapping composition is general, i.e., mappings from LVE and non-LVE transformations are still composable. Hence, our framework can be extended with algorithms that generate mappings for other transformations (e.g., vectorization-based ones) and the compensation code they produce can be combined with the one from LVE transformations. Function transformations such as inlining would instead require extending our formalism to account for procedures and for the relations between points across functions.

We would like to remark that the assumption of an identity mapping between program points required for live-variable bisimilarity is without loss of generality. In fact, it can always be enforced by padding programs with \wskip\ statements (e.g., the Hoist rule in \myfigure\ref{fig:sample-trans} expects a \wskip\ to already exist at the point where an instruction is moved) and is not required in a real compiler as we will see in \mysection\ref{ss:tracking-opt}. 


\section{LLVM Implementation}
\label{se:implementation}

In this section we present and evaluate an implementation in LLVM of our techniques for automatic OSR mapping construction. In particular, we discuss how to deal with the presence of memory \load\ and \store\ instructions, and how to implement algorithms \apply\ and \buildcomp\ in a real compiler.
We then investigate whether in the presence of a number of common compiler optimizations, \buildcomp\ can offer an extensive ``menu'' of possible program points where OSR can safely occur, generating the possibly required compensation code in an automated fashion. Our experiments suggest that bidirectional OSR is supported almost everywhere in this setting.

\subsection{The LLVM Compiler Infrastructure}
LLVM is designed to support transparent, life-long program analysis and transformation for arbitrary programs~\cite{Lattner04}. Front-ends are available for a number of static languages (e.g., \clang\ for C, C++, and Objective C/C++), while its MCJIT just-in-time compiler is currently employed to generate optimized code in virtual machines for a variety of dynamic languages, including Python, Ruby, Julia, and R.


The core of LLVM is its low-level intermediate representation (IR). A high-level language front-end  compiles a program's source code to LLVM IR; platform-independent {\em optimization passes} manipulate the IR, and a back-end eventually compiles it to native code, performing architecture-specific optimizations such as register allocation. A shared extensive optimization pipeline is offered to front-end authors to generate efficient code for their language. 

LLVM provides an infinite set of typed {\em virtual registers} in static single assignment (SSA) form~\cite{Cytron91}, and values can be transferred between registers and memory solely via \load\ and \store\ operations. When a program variable might assume a different value depending on where the control flow came from, a $\phi$ function merges multiple incoming virtual registers into a new one, i.e., a $\phi$-node.
Front-ends do not have to generate code in SSA form: they can place variables on the stack using the \alloca\ instruction, and access them using \load\ and \store. The \memtoreg\ pass will then construct the SSA form by promoting stack references to virtual registers.

\vspace{-1mm} 
\subsection{Integration with \texorpdfstring{\texttt{OSRKit}}{OSRKit}}
\label{ss:osrkit}

\osrkit~\cite{DElia16} is an LLVM library working at IR level: it allows a front-end to perform OSR at arbitrary locations, provided that optimizers can generate code to realign the state after the transition. This library overcomes limitations of previous OSR work in LLVM~\cite{Lameed13} that provides support for transitions at loop headers only when no state adjustments are required.

Given a base function \fbase, a variant \fvariant\ to ``OSR'' into, and a location \osrsource\ in \fbase, \osrkit\ instruments \fbase\ with an OSR point guarded by a user-provided condition. The transition is modeled as a function call that transfers the live state to a newly generated continuation function \fosrto, which is an efficient, specialized version of \fvariant\ that executes any required compensation code at its entry point before jumping to the resumption point \osrlanding.

\cite{DElia16} focuses on the engineering aspects for supporting OSR with compensation code in LLVM, presenting a case study on dynamic inlining with aggressive type specialization in MATLAB in which compensation code is hand-written. This article makes a step forward showing how to automatically generate and compose compensation code for LVE transformations on top of \osrkit\ using the algorithms from \mysection\ref{ss:osr-mapping-algorithms}.

\subsection{Making Existing LLVM Passes OSR-Aware}
\label{ss:tracking-opt}
In this section, we discuss how to make existing LVE LLVM optimization passes OSR-aware. For the sake of simplicity, in Section~\ref{se:framework} we have made the impractical assumption that an OSR always jumps from a program point $l$ in $\pi$ to the same program point $l'=l$ in $\pi'$. However, in a real implementation a mapping between LLVM instruction locations across versions should be explicitly created by suitably defining the \apply\ function of \mysection\ref{ss:osr-mapping-algorithms}. We notice that it is sufficient to instrument LLVM optimizations at places where IR manipulations are done without having to rewrite them. We recall that, once the mappings $\Delta$ and $\Delta'$ between program points are created by \apply, compensation code can be automatically constructed using \buildcomp\ (\myalgorithm\ref{alg:osr-build-comp}).

Without loss of generality, we can capture the effects of an LVE program transformation in terms of six primitive actions: (1) \mytt{add}$(inst, loc)$ and (2) \mytt{delete}$(inst)$ to model code insertion and deletion; (3) \mytt{hoist}$(loc, newLoc)$ and (4) \mytt{sink}$(loc, newLoc)$ to move instructions; (5) \mytt{replace}$(inst, oldOp, newOp)$ to modify an operand of an instruction; and (6) \mytt{replaceAll}$(oldOp, newOp)$ to replace an operand with another in all of its uses in the function.


\oldrevision{
\begin{enumerate}[parsep=0pt,partopsep=0pt]
 \item \mytt{add}$(inst, loc)$: insert a new instruction $inst$ at location $loc$;
 \item \mytt{delete}$(loc)$: delete the instruction at location $loc$;
 \item \mytt{hoist}$(loc, newLoc)$: hoist an instruction from location $loc$ to $newLoc$;
 \item \mytt{sink}$(loc, newLoc)$: sink an instruction from location $loc$ to $newLoc$;
 \item \mytt{replace\_operand}$(inst, old\_op, new\_op)$: replace an operand $old\_op$ for a a given instruction $inst$ with another operand $new\_op$;
 \item \mytt{replace\_all}$(old\_op, new\_op)$: replace all uses of an operand in the code with another operand.
\end{enumerate}
}

Our implementation of \apply\ takes as input a function and an optimization, clones the function, optimizes the clone, and eventually constructs a mapping between program points in the two versions by processing the history of applied actions. The mapping is augmented with information correlating virtual registers from the two functions when fresh IR objects are introduced, e.g., an instruction is replaced with a more efficient one.
In our experience, to make an LLVM pass OSR-aware we had to insert 5-15 tracking primitive actions. The hardest part was clearly understanding what each LLVM pass does. Readers familiar with LLVM's internals may notice that most primitive actions mirror typical manipulation utilities used in optimization passes. 

\oldrevision{
Tracking actions of the first four kinds is essential in order to correlate program points, as a mapping between LLVM instruction locations across versions should be explicitly maintained.
An OSR mapping between two LLVM functions is defined in terms of virtual registers. A \RAUWfull\ operation updates the mapping as follows. When all uses of $O$ are replaced with $N$, $O$ becomes trivially dead: as in LVB programs $N$ and $O$ yield the same result, any virtual register $O'$ in the mapping pointing to $O$ can be updated to point to $N$. This is useful for deoptimization, as our experiments suggest that in an optimized program version a single variable can often be used in place of multiple variables from the unoptimized version.
}

\subsection{Supporting \texorpdfstring{\texttt{load}}{load} and \texorpdfstring{\texttt{store}}{store} Instructions}
\label{ss:load-store}
LLVM provides \load\ and \store\ instructions to transfer values between memory and virtual registers.
A simple sufficient condition for multi-program determinism is that \store\ instructions are executed at the same program point in all versions. Indeed, when two program versions assign to a variable with a \load\ from the same address, and the variable is live at some same program point in both versions, then the value read from memory has to be the same in both versions. Our implementation preserves the \store\ invariant above while allowing instructions that do not access memory to be hoisted above or sunk below a \store\ instruction. Common LLVM optimizations such as loop hoisting and code sinking deal with \store\ instructions in a similar manner.

A possible extension for scenarios where the above assumption might be too restrictive is as follows. Suppose that a \store\ is sunk during optimization. For each CFG location between the original location and the insertion point: (a) in an OSR to the optimized version, no compensation code is required, as the \store\ has been executed already, and re-executing it at the insertion point will be harmless; (b) in an OSR to the base version, we have to realign the memory state by executing the sunk \store, which has not been reached yet in the optimized version.

\subsection{Implementing \texorpdfstring{\mytt{build\_comp}}{build\_comp} and \texorpdfstring{\mytt{reconstruct}}{reconstruct}}
\label{ss:buildcomp-implementation}
We now discuss the implications of implementing \buildcomp\ (\myalgorithm\ref{alg:osr-build-comp}) for programs in SSA form. While this form guarantees that the reaching definition for a variable is unique at any point it dominates, \reconstruct\ gives up when attempting to reconstruct an assignment made through a $\phi$ function. Our current implementation also conservatively prevents \reconstruct\ from inserting \load\ instructions in the compensation code.

Compared to the abstract model described in \mysection\ref{se:language-framework}, the particular form of IR code generated by LLVM may limit the effectiveness of \reconstruct\ in our context. We have thus implemented three versions of the algorithm. We denote by $P$ the pool of variables at the OSR source that can be used to reconstruct the assignments. The $live$ version is the base version of \myalgorithm\ref{alg:osr-reconstruct} that includes in $P$ only those variables that are live at the OSR source.

The $live_{opt}$ version has a few enhancements. It can recursively reconstruct constant $\phi$-assignments\footnote{A constant $\phi$-assignment merges together the same value for all CFG paths. Examples are $\phi$-nodes placed by compilers at loop exits for values that are live across the loop boundary when constructing the so-called {\em Loop-Closed} SSA (LCSSA) form.} and includes in $P$ also non-live function parameters, as arguments cannot be modified by IR instructions in LLVM. $live_{opt}$ also exploits implicit aliasing information deriving from a \RAUW$(O,$ $N)$, as the corresponding $O'$ variable for $O$ in the mapping can be used to reconstruct $N$ when $N'$ is not live at the OSR source location. In fact, in an optimizing OSR a variable to set at the destination might  be aliased by multiple variables at the source.

\oldrevision{
Results for the $alias$ version of \reconstruct\ are not reported as they do not improve over $live_{(e)}$. Indeed, aliasing information is useful when a variable to set at the destination is aliased by multiple variables at the source, which we do not expect to happen after optimizations such as CSE have been applied. 
}

The $avail$ version includes in $P$ also those virtual registers that are not live at the source location, but contain {\em available} values that \reconstruct\ can directly assign to the instruction operand (line~\ref{line-rec:recursive}) or assignment (line~\ref{line-rec:assign}) being reconstructed. We exploit the uniqueness of reaching definitions to efficiently identify such variables. 

\oldrevision{
\begin{enumerate}
 \item The first version, which we will refer to as $live$, is the base version of \myalgorithm\ref{alg:osr-reconstruct} that includes in $P$ only those variables that are live at the OSR source.
 \item An enhanced version $live_{(e)}$ exploits some unique features of LLVM IR. In particular, it can recursively reconstruct a $\phi$-assignment that merges together the same value for all CFG paths\footnote{Compilers can place $\phi$-nodes at loop exits for values that are live across the loop boundary, constructing the so-called {\em Loop-Closed} SSA (LCSSA) form.}, and extends $P$ by including also non-live function parameters, as arguments cannot be modified by IR instructions.
 \item A third version, which we call $alias$, can exploit implicit aliasing information deriving from a \RAUW$(O,$ $N)$. Let $O'$ and $N'$ be the corresponding variables according to the OSR mapping for $O$ and $N$, respectively: we can add \mytt{N:=O'} to the compensation code required to reconstruct $N$ when $N'$ is not live at the source location, but $O'$ is.
 \item Finally, the fourth version $avail$ includes in $P$ also those virtual registers that are not live at the source location, but contain {\em available} values that \reconstruct\ can directly assign to the instruction operand (line~\ref{line-rec:recursive}) or assignment (line~\ref{line-rec:assign}) being reconstructed. We exploit the uniqueness of reaching definitions to efficiently identify such variables. 
\end{enumerate}
}



\subsection{Evaluation}
\label{ss:evaluation}

In this section we show that our algorithms enable bidirectional OSR transitions on prominent benchmarks almost everywhere in the code across multiple common, unhindered compiler optimizations.


\paragraph{Benchmarks and Environment}
We integrate our techniques in \tinyvm, a proof-of-concept virtual machine that provides an interactive environment for LLVM IR manipulation, JIT compilation, and benchmarking~\cite{DElia16AnonymousThesis}. We extend \tinyvm\ to automatically construct and compose OSR mappings for a sequence of transformations applied to a function $f_{base}$ to generate an optimized version $f_{opt}$. For each feasible OSR point in $f_{base}$/$f_{opt}$, we invoke \osrkit\ to materialize the compensation code $\chi$ produced by \reconstruct\ into a sequence of IR instructions for the OSR entry block of the continuation function $f_{opt_{to}}$/$f_{base_{to}}$ (\mysection\ref{ss:osrkit}).

We evaluate our technique on the \speccpu~\cite{Henning06} and the \phoronixpts~\cite{Phoronix16} benchmarking suites, reporting data for a subset of their C/C++ benchmarks. We profile each benchmark to identify the hottest method and when it accounts for at least 5\% of the total execution time, we pick it, generating its IR using \clang. No optimization is enabled during the compilation other than \memtoreg. Starting from this IR version, henceforth {\em base}, we generate an {\em opt} version by applying all the optimizations we discuss next.
We run our experiments on an Intel Core i7-3632QM machine running Ubuntu 14.10 (64 bit) and LLVM 3.6.2.

\paragraph{Optimizations}
\newcommand{\optnames}[1]{#1} 
We instrument a number of standard LLVM optimization passes, including \optnames{aggressive dead code elimination} (ADCE), \optnames{constant propagation} (CP), \optnames{common subexpression elimination} (CSE), \optnames{loop-invariant code motion} (LICM), \optnames{sparse conditional constant propagation} (SCCP), and \optnames{code sinking} (Sink). We also instrument utility passes required by LICM such as \optnames{natural loop canonicalization} (LC) and \optnames{LCSSA-form construction} (LCSSA). Notice that optimizations performed by the back-end such as instruction scheduling and register allocation do not require instrumentation, as we operate at IR level. 

\mytable\ref{tab:OSR-alC-bench-IR} shows aggregate figures for IR manipulations performed by the optimizations on our benchmarks. Reported numbers suggest that while the {\em opt} version is typically shorter than its {\em base} counterpart, it might have a larger number of $\phi$-nodes: most extra nodes are commonly generated during the LCSSA-form construction and eventually optimized away in the back-end. SCCP can eliminate a large number of unreachable basic blocks for \mytt{ffmpeg}, while for the remaining benchmarks the majority of instruction deletions are performed by CSE.

\begin{table}[t]
\vspace{-1mm}
\begin{small}
\caption{\label{tab:OSR-alC-bench-IR} IR features of hottest function in each benchmark. We report the number of instructions $|\pi|$ ($|\phi|$ of which represent $\phi$-nodes) for  the {\em base} and the {\em opt} version, along with the number of primitive code manipulation actions tracked during optimization.  $R_{\{A,C,I\}}$ stands for \RAUWfull\ actions for some {$N$} of LLVM {\em Argument}, {\em Constant}, or {\em Instruction} type.}
\makebox[\textwidth][c]{
\begin{adjustbox}{width=1.25\textwidth}
\begin{footnotesize}
\begin{tabular}{|c|p{3.8cm}|c|c|c|c|c|c|c|c|c|c|c|}
\cline{3-6}
\multicolumn{2}{l|}{} & \multicolumn{2}{c|}{base} & \multicolumn{2}{c|}{opt} \\
\hline
Benchmark & Function & $|\pi|$ & $|\phi|$ & $|\pi|$ & $|\phi|$  & \texttt{add} & \texttt{delete} & \texttt{hoist} & \texttt{sink} & $R_I$ & $R_C$ & $R_A$ \\
\hline
\hline
bzip2 & mainSort & 657 & 32 & 596 & 44 & 16 & 77 & 12 & 3 & 71 & 0 & 2 \\
\hline
h264ref & SetupFastFullPelSearch & 671 & 28 & 576 & 36 & 9 & 105 & 4 & 21 & 102 & 0 & 0 \\
\hline
hmmer & P7Viterbi & 568 & 6 & 383 & 8 & 2 & 187 & 13 & 1 & 187 & 0 & 0 \\
\hline
namd & \seqsplit{ComputeNonbondedUtil::calc\_pair\_energy\_fullelect} & 1737 & 159 & 1636 & 224 & 68 & 169 & 36 & 73 & 145 & 17 & 0 \\
\hline
perlbench & S\_regmatch & 5574 & 305 & 5001 & 355 & 86 & 667 & 96 & 28 & 627 & 0 & 0 \\
\hline
sjeng & std\_eval & 1940 & 93 & 1540 & 105 & 13 & 413 & 20 & 34 & 412 & 1 & 0 \\
\hline
soplex & SPxSteepPR::entered4X & 195 & 2 & 154 & 2 & 0 & 41 & 2 & 4 & 41 & 0 & 0 \\
\hline
bullet & \seqsplit{btGjkPairDetector::getClosestPointsNonVirtual} & 587 & 24 & 553 & 42 & 26 & 60 & 37 & 3 & 51 & 1 & 0 \\
\hline
dcraw & vng\_interpolate & 590 & 37 & 545 & 49 & 13 & 58 & 25 & 6 & 58 & 0 & 0 \\
\hline
ffmpeg & decode\_cabac\_residual\_internal & 618 & 34 & 462 & 40 & 11 & 168 & 9 & 17 & 52 & 51 & 0 \\
\hline
fhourstones & ab & 288 & 29 & 284 & 39 & 14 & 20 & 3 & 0 & 14 & 2 & 0 \\
\hline
vp8 & vp8\_full\_search\_sadx8 & 334 & 41 & 299 & 60 & 19 & 54 & 17 & 34 & 54 & 0 & 0 \\
\hline
\end{tabular}
\end{footnotesize}
\end{adjustbox}
}
\end{small}
\vspace{-1mm}
\end{table}


\paragraph{Optimizing OSR}

\myfigure\ref{fig:BtoO} shows the fraction of program points that are feasible for an OSR from {\em base} to {\em opt} depending on the version of \reconstruct\ in use. 
Locations that can fire an OSR with no need for a compensation code (i.e., $\chi=\langle\rangle$) account for a limited fraction of all the potential OSR points (less than $10\%$ for most benchmarks). This suggests that optimizations can significantly modify a program's live state across program locations. 

\begin{figure}[t]
\begin{center}
\makebox[\textwidth][c]{
\begin{subfigure}{0.585\textwidth}
\includegraphics[height=4.2cm]{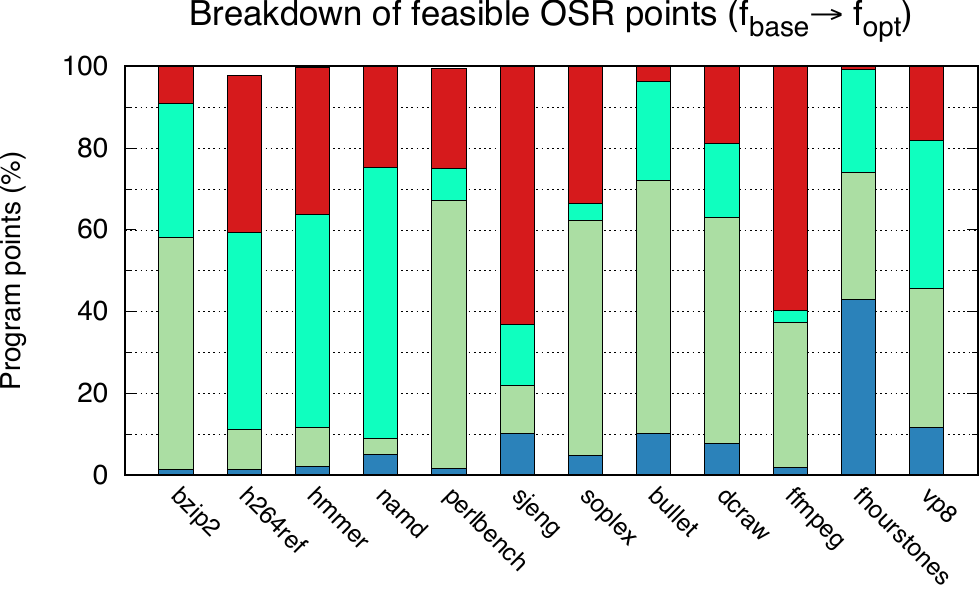}
\caption{\label{fig:BtoO}}
\end{subfigure}
\begin{subfigure}{0.665\textwidth}
\includegraphics[height=4.2cm]{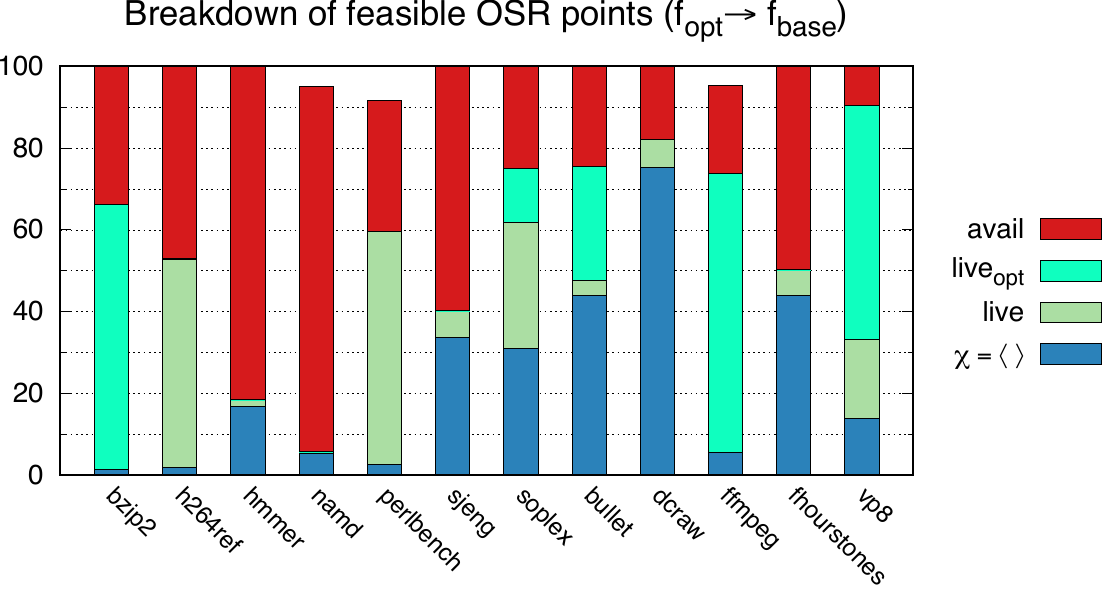}
\caption{\label{fig:OtoB}\hphantom{xxxxxxxxxxx}}
\end{subfigure}
}
\vspace{-2mm}
\caption{Fraction of program points that are OSR-feasible: (a) from {\em base} to {\em opt}, and (b) from {\em opt} to {\em base}.}
\end{center}
\vspace{-2mm}
\end{figure}

We observe that $live$ performs well on some benchmarks (e.g., \mytt{perlbench}, \mytt{bullet}, \mytt{dcraw}) and poorly on others (e.g., \mytt{h264ref}, \mytt{namd}). The enhancements introduced in $live_{opt}$ increase the number of feasible OSR points for all benchmarks. For $9$ out of $12$ of them, it becomes possible to build a compensation code using only live variables at the OSR source for more than $60\%$ of potential OSR points.

When in the $avail$ version \reconstruct\ is allowed to extend the liveness range of an available variable (i.e., an already-evaluated virtual register), the percentage of feasible OSR points grows to nearly $100\%$. We observe for \mytt{bullet} that the same $\phi$-node needs to be reconstructed at nearly $20\%$ of feasible OSR points: this node takes as incoming values a number of $\phi$-nodes that in turn all yield the same available value. Differently than LLVM's built-in method for detecting constant $\phi$-nodes, our recursive heuristic can correctly identify and use such value.


In \mytable\ref{tab:OSR-alC-prologue} we report the average and peak size of the compensation code $\chi$ generated by $live_{opt}$ and $avail$ across feasible OSR points. Figures for $live$ would add little to the discussion and are not reported. Notice that average values are calculated on different sets of program points, as $avail$ extends the set from $live_{opt}$.


The assignment step of \reconstruct\ (line~\ref{line-rec:assign}) generates an average number of instructions typically smaller than $20$, with the notable exception of \mytt{perlbench}. Its hottest function highly benefits from CSE: we found out that no less than $583$ out of its $667$ deleted instructions ($\approx10\%$ of the {\em base} function size) are removed by it. We believe that local CSE would shrink the OSR entry block of the continuation function $f'$ as well. However, this optimization is not strictly necessary. The size of $\phi$ is unlikely to affect the performance of $f'$ for a hot method, as compensation code will be located at the beginning of the continuation function and executed only once.

\mytable\ref{tab:OSR-alC-prologue} also reports the average and peak number of variables ($|K_{avail}|$) that are not live at the source location, but for which $avail$ would artificially extend liveness to support OSR at the program points represented by the top bars in \myfigure\ref{fig:BtoO}. We observe that the average number of values to keep alive is less than $3$ for $9$ out of $12$ benchmarks, with a maximum of $6.15$ for \mytt{bullet}. By using a simple backtracking strategy, $avail$ extends the liveness of an available value only when it is impossible to reconstruct it otherwise.


\begin{table}[!ht]
\begin{small}
\caption{\label{tab:OSR-alC-prologue} Average and peak size $|\chi|$ of the compensation code generated by the $live_{opt}$ and $avail$ versions of procedure \reconstruct. $|K_{avail}|$ is the size of the set of variables that we should artificially keep alive in order to allow an OSR from {\em base} to {\em opt} at program points represented by the top bars in \myfigure\ref{fig:BtoO} and \myfigure\ref{fig:OtoB}.}
\vspace{-1mm}
\makebox[\textwidth][c]{
\begin{adjustbox}{width=1.25\textwidth}
\begin{footnotesize}
\begin{tabular}{ |c|c|c|c|c|c|c|c|c|c|c|c|c| }
\cline{2-13}
\multicolumn{1}{l|}{} & \multicolumn{6}{c|}{$f_{base}\rightarrow f_{opt}$} & \multicolumn{6}{c|}{$f_{opt}\rightarrow f_{base}$} \\
\cline{2-13}
\multicolumn{1}{l|}{} & \multicolumn{2}{c|}{$|\chi|\leftarrow live_{opt}$} & \multicolumn{2}{c|}{$|\chi|\leftarrow avail$} & \multicolumn{2}{c|}{$|K_{avail}|$} & \multicolumn{2}{c|}{$|\chi|\leftarrow live_{opt}$} & \multicolumn{2}{c|}{$|\chi|\leftarrow avail$} & \multicolumn{2}{c|}{$|K_{avail}|$} \\
\hline
Benchmark & Avg & Max & Avg & Max & Avg & Max & Avg & Max & Avg & Max & Avg & Max \\
\hline
\hline
bzip2 & 4.3 & 14 & 4.73 & 13 & 3.6 & 8 & 1.55 & 4 & 1.77 & 4 & 1.47 & 4 \\
\hline
h264ref & 2.9 & 5 & 3.37 & 5 & 1.02 & 2 & 4.46 & 9 & 2.82 & 9 & 1.45 & 7 \\
\hline
hmmer & 16.11 & 23 & 16.63 & 24 & 4.02 & 7 & 1 & 1 & 1 & 1 & 1.02 & 2 \\
\hline
namd & 18.61 & 28 & 17.82 & 28 & 3.38 & 6 & 1.5 & 2 & 5.93 & 15 & 4.74 & 18 \\
\hline
perlbench & 46.12 & 57 & 45.82 & 57 & 1.24 & 12 & 4.09 & 12 & 4.22 & 12 & 1.37 & 11 \\
\hline
sjeng & 9.72 & 21 & 18.52 & 32 & 4.2 & 12 & 1.29 & 2 & 1.67 & 11 & 4.09 & 14 \\
\hline
soplex & 5.02 & 7 & 4.38 & 7 & 2.34 & 4 & 3.3 & 4 & 3.3 & 4 & 1.00 & 1 \\
\hline
bullet & 16.69 & 46 & 15.93 & 46 & 6.15 & 17 & 1 & 1 & 1.26 & 3 & 1.14 & 2 \\
\hline
dcraw & 7.6 & 15 & 7.32 & 15 & 1.97 & 7 & 1.68 & 2 & 3.84 & 6 & 4.06 & 8 \\
\hline
ffmpeg & 5.05 & 8 & 4.03 & 8 & 1.85 & 3 & 1.94 & 5 & 1.95 & 6 & 1.08 & 4 \\
\hline
fhourstones & 4.5 & 6 & 4.98 & 6 & 1.7 & 2 & 0 & 0 & 1.12 & 4 & 1.42 & 4 \\
\hline
vp8 & 10.51 & 16 & 10.13 & 17 & 2.35 & 6 & 5.74 & 13 & 5.51 & 13 & 1.18 & 5 \\
\hline
\hline
Avg & {\bf 12.26} & 20.50 & {\bf 12.81} & 21.50 & {\bf 2.82} & 7.17 & {\bf 2.30} & 4.58 & {\bf 2.87} & 7.33 & {\bf 2.00} & 6.67 \\
\hline
\end{tabular}
\end{footnotesize}
\end{adjustbox}
}
\end{small}
\end{table}

\paragraph{Deoptimizing OSR}

\myfigure\ref{fig:OtoB} reports the fraction of OSR points eligible for {\em opt}-to-{\em base} deoptimization. We observe that the fraction of locations that can fire an OSR with an empty $\chi$ varies significantly from benchmark to benchmark, suggesting a dependence on the structure of the original program.

For $9$ out of $12$ benchmarks, compensation code can be built using only live variables for more than $50\%$ of potential OSR points.
When the $avail$ version is used, the percentage of feasible OSR points is greater than $90\%$ on all benchmarks and nearly $100\%$ for $9$ out of $12$ of them.
In \mytable\ref{tab:OSR-alC-prologue} we then report statistics about the size of the compensation code generated across feasible OSR points, and the number of available variables to be kept alive in $avail$. Compared to the optimizing OSR scenario, the size of the compensation code is much smaller, suggesting that shorter portions of execution need to be reconstructed in a deoptimizing OSR. 

Note that the $0$ values reported for \mytt{fhourstones} in the $live_{opt}$ scenario do not mean that state compensation is not required. In fact, the algorithm detects that each variable $x$ to be rematerialized at the OSR landing pad is aliased by either a non-live function argument or a live constant $\phi$-node. All the uses of $x$  in the code can thus be replaced by uses of the alias when generating the OSR continuation function.


\subsection{Discussion}
\label{ss:experiments-discussion}
Our LLVM implementation requires an OSR mapping to be maintained between the original and the optimized version of a function. Runtime guards inserted by \osrkit\ are transparent to it, and a specialized continuation function generated for the OSR landing pad will resume the execution at full speed
~\cite{Fink03}.

We have seen that common compiler transformations can significantly affect the live state of a program across its locations. The three versions of \reconstruct\ we have implemented can generate compensation code automatically by recursively reassembling portions of the state for the target function. 
%
OSR is supported at more than a half of the program locations by $live_{opt}$, and almost everywhere by $avail$. Figures reported in \mytable\ref{tab:OSR-alC-prologue} suggest that the size of the set of virtual registers to preserve for an OSR point enabled only by $avail$ is small.


We remark that extending the liveness range of an available virtual register $r$ should not be an issue in terms of register pressure increase. If $r$ is assigned to a physical register, a compiler would normally spill it to the stack before it gets clobbered, to only reload it later when an OSR is about to be fired. If $r$ is assigned to a stack location instead, it should be loaded to a physical register only when an OSR is performed. In both cases, we would never reload a register more than once. Furthermore, \osrkit\ allows a front-end to encode the probability of an OSR transition in terms of control-flow edge weights to guide native code generation. Keeping an otherwise dead value in a register makes sense only when used at an OSR point that is very likely to be fired.

\section{Case Study: Source-Level Debugging of Optimized Code}
\label{se:debugging}

In this section we present a case study that shows how our algorithms for compensation code generation can provide useful novel building blocks for optimized-code debuggers. On prominent C benchmarks, \reconstruct\ is able to recover the expected source-level values for the vast majority of scalar user variables that might not be reported correctly by a debugger due to the effects of classic compiler optimizations.

\subsection{Background}
\label{ss:debugging-intro}

A {\em source-level} (or {\em symbolic}) {\em debugger} is a program development tool that allows a programmer to monitor an executing program at the source-language level. Interactive mechanisms are typically provided to the user to halt/resume the execution at {\em breakpoints}, and to inspect the state of the program in terms of its source language.

The importance of the design and use of these tools was already clear in the '60s~\cite{Evans66}. In a production environment it is desirable to use optimizations, as bugs can surface when they are enabled (a debuggable translation of a program may hide bugs) or because differences in timing behavior may cause the appearance of bugs due to race conditions. Also, optimizations may be absolutely necessary to execute a program due to memory limitations, efficiency reasons, or other platform-specific constraints~\cite{Adl-Tabatabai96thesis}.

As pointed out by Hennessy in a seminal work~\cite{Hennessy82}, a classic conflict exists between the use of optimization techniques and the ability to debug a program symbolically. A debugger provides the user with the illusion that the source program is executing one statement at a time. Optimizations preserve semantic equivalence between the executed and the original code, but normally alter the structure and the intermediate results of the program.

Two problems surface when trying to symbolically debug optimized code~\cite{Adl-Tabatabai96,Jaramillo00}. First, the debugger must determine the position in the optimized code that corresponds to a breakpoint ({\em code location} problem). Second, the user expects to see the values of source variables at a breakpoint in a manner consistent with the source code, even though the optimizer might have 
deleted or reordered instructions, or values might have been overwritten as a consequence of register allocation choices ({\em data location} problem).

When attempting to debug optimized programs, debuggers may thus give misleading information about the value of variables at breakpoints. Hence, the programmer has the difficult task of attempting to unravel the optimized code and determine what values the variables should have~\cite{Hennessy82}. When global optimizations can cause the run-time value of a variable to be inconsistent with the source-level value expected at the breakpoint, the variable is called {\em endangered}~\cite{Adl-Tabatabai96}.

In general, for a symbolic debugger there are two ways to present meaningful information about the debugged optimized program~\cite{Wu99}. It can provide {\em expected behavior} of the program when it hides the effects of the optimizations from the user and presents the program state consistent with what they expect from the unoptimized code. It provides instead {\em truthful behavior} if it makes the user aware of the effects of the optimizations and warns them of possibly surprising outcomes.
\cite{Adl-Tabatabai96thesis} observes that constraining optimizations or adding machinery during compilation to aid debugging does not solve the problem of debugging the optimized translation of a program, as the user debugs suboptimal code. Source-level debuggers should thus explore techniques to recover expected behavior without relying on intrusive compiler extensions.

\subsection{Using \texorpdfstring{$\texttt{reconstruct}$}{reconstruct} for State Recovery}
\label{ss:buildcomp-for-recovery}

On-stack replacement has been pioneered in implementations of the SELF programming language to provide expected behavior with globally optimized code~\cite{Holzle92}. OSR can shield a debugger from the effects of optimizations by dynamically deoptimizing code on demand. Debugging information is supplied by the compiler at discrete {\em interrupt points}, which act as a barrier for optimizations, letting the compiler run unhindered between them. Motivated by the observation that our algorithms for generating OSR mappings do not place such restrictions on LVE transformations and can be applied at any program location, we investigate whether they can also encode useful information for providing expected behavior in a source-level debugger. 


As in most recent works on optimized code debugging, we focus on identifying and recovering {\em scalar source variables} in the presence of global optimizations. In LLVM, debugging information is inserted at IR level by the front-end as {\em metadata} attached to global variables, single instructions, functions or entire modules.
These metadata are transparent to optimization passes, they do not prevent them from happening, and are designed to be agnostic about both the source language behind the original program and the target debugging information representation. Two intrinsics associate IR objects with source-level variables: \mytt{llvm.dbg.declare} associates a variable with the address of an \alloca\ buffer; \mytt{llvm.dbg.value} associates a variable with the content of a register.
\oldrevision{
\begin{itemize} 
 \item \mytt{llvm.dbg.declare} associates a variable with the address of an \alloca\ buffer;
 \item \mytt{llvm.dbg.value} tracks that a variable is being assigned with the content of a register.
\end{itemize}
}

We extend \tinyvm\ to reconstruct this mapping and identify which locations in the unoptimized IR $f_{base}$ correspond to source-level locations (i.e., possible breakpoint locations) for a function. An OSR mapping is constructed when LVE transformations are applied to $f_{base}$ to generate $f_{opt}$. For each location in $f_{opt}$ that might correspond to (i.e., have as OSR landing pad) a source-level location in $f_{base}$, we determine which live variables at the destination are live also at the source (and thus yield the same value), and which ones need to be reconstructed instead.
We rely on the SSA form to identify which assignments should be recovered, as every value instance for a source-level variable is represented by a specific virtual register. $\phi$-nodes at control-flow merge points cannot be reconstructed, but our experimental results suggest that this might not be a common issue in practice.


\subsection{The \texorpdfstring{$\speccpu$}{SPEC CPU2006} Benchmarks}
\label{ss:spec-benchmarks} 

To capture a variety of programming patterns and styles from applications with different sizes, we analyze each method of each C benchmark from the \speccpu\ suite, 
applying the same sequence of OSR-aware optimization passes as in \mysection\ref{ss:evaluation} to the baseline IR version obtained with \clang\ \mytt{-O0} 
followed by \memtoreg.
\mytable\ref{tab:CS-debug-benchmarks} reports for each benchmark the code size (LOC), the total number of functions in it ($|F_{tot}|$), the number of functions modified by the applied optimizations ($|F_{opt}|$) and, in turn, how many optimized functions are {\em endangered} ($|F_{end}|$), i.e., contain endangered user variables and may require recovery of the expected behavior.

We observe that $11\%$ (\mytt{libquantum}) to $54\%$ (\mytt{gobmk}) of the optimized functions are endangered, while for
$10\%$ to $33\%$ of the functions in each benchmark, the applied IR-level optimizations do not kick in.
%
%
%
For endangered functions, on average at more than $25\%$ of program points there is at least a user variable whose source-level value might not be reported correctly by a debugger. For most functions in the benchmarks, the average number of affected user variables at such points ranges between $1$ and $2$, although for some benchmarks we observe higher peaks at specific program locations (e.g., as high as $9$ for \mytt{gobmk} and $14$ for \mytt{gcc} and \mytt{h264ref}). 

To investigate possible correlations between the size of a function and the number of user variables affected by source-level debugging issues, we analyze the corpus of functions for the three largest benchmarks in our suite, i.e., \mytt{gcc}, \mytt{gobmk}, and \mytt{perlbench}. Our findings (Appendix~\ref{apx:additional-material}) suggest that although larger functions might be more prone to have a large number of affected variables, such issues frequently arise for smaller functions as well.

\subsection{Experimental Results}
\label{ss:debugging-results}


We evaluate the ability of \reconstruct\ to correctly recover the source-level expected value for endangered user variables in the \speccpu\ experiments. For each function, we measure the {\em average recoverability ratio}, defined as the average across all program points corresponding to source-level locations of the ratio between recoverable and endangered user variables at each point. Two versions of \reconstruct\ can be used here.

\begin{table}[t]
\begin{small}
\caption{\label{tab:CS-debug-benchmarks} \speccpu\ C benchmarks suite: for endangered functions, we report weighted $Avg_g$ and unweighted $Avg_u$ average of the fraction of program points with endangered user variables, then mean, standard deviation, and peak number of endangered variables at such points. We use the number of IR instructions in the unoptimized code as weight for computing $Avg_w$, and consider only IR program points corresponding to source-level locations.
}
\makebox[\textwidth][c]{
\begin{adjustbox}{width=1.25\textwidth}
\begin{footnotesize}
\begin{tabular}{ |c|r|r|r|r|r|r|C{0.9cm}|C{0.9cm}|C{0.6cm}|C{0.6cm}|r| }
\cline{8-12}
\multicolumn{7}{l|}{} & \multicolumn{5}{c|}{\em Endangered functions} \\
\cline{3-12}
\multicolumn{2}{l|}{} & \multicolumn{5}{c|}{\em Functions} & \multicolumn{2}{c|}{Fraction of affected} & \multicolumn{3}{c|}{Endangered user vars} \\ 
\cline{3-7}
\multicolumn{2}{l|}{} & \multicolumn{1}{c|}{Total} & \multicolumn{2}{c|}{IR Optimized} & \multicolumn{2}{c|}{Endangered} & \multicolumn{2}{c|}{program points} & \multicolumn{3}{c|}{per affected point} \\
\hline
Benchmark & \multicolumn{1}{c|}{LOC} & \multicolumn{1}{c|}{$|F_{tot}|$} & \multicolumn{1}{c|}{$|F_{opt}|$}  & \multicolumn{1}{c|}{\tiny$\frac{|F_{opt}|}{|F_{tot}|}$} & \multicolumn{1}{c|}{$|F_{end}|$} & \multicolumn{1}{c|}{\tiny$\frac{|F_{end}|}{|F_{opt}|}$} & $Avg_w$ & $Avg_u$ & $Avg$ & $\sigma$ & $Max$ \\ 
\hline
\hline
bzip2 & 8\,293 & 100 & 66 & 0.66 & 24 & 0.36 & 0.17 & 0.12 & 1.22 & 0.55 & 5 \\ 
\hline
gcc & 521\,078 & 5\,577 & 3\,884 & 0.70 & 1\,149 & 0.30 & 0.25 & 0.22 & 1.13 & 0.31 & 14 \\
\hline
gobmk & 197\,215 & 2\,523 & 1\,664 & 0.66 & 893 & 0.54 & 0.40 & 0.29 & 1.48 & 0.72 & 9 \\ 
\hline
h264ref & 51\,578 & 590 & 466 & 0.79 & 163 & 0.35 & 0.45 & 0.55 & 1.69 & 1.23 & 14 \\ 
\hline
hmmer & 35\,992 & 538 & 429 & 0.80 & 80 & 0.19 & 0.17 & 0.22 & 1.13 & 0.37 & 5 \\ 
\hline
lbm & 1\,155 & 19 & 17 & 0.89 & 2 & 0.12 & 0.30 & 0.51 & 1.97 & 1.37 & 3 \\ 
\hline
libquantum & 4\,358 & 115 & 85 & 0.74 & 9 & 0.11 & 0.13 & 0.10 & 1.06 & 0.17 & 2 \\ 
\hline
mcf & 2\,658 & 24 & 21 & 0.88 & 11 & 0.52 & 0.35 & 0.32 & 1.00 & - & 1 \\ 
\hline
milc & 15\,042 & 235 & 157 & 0.67 & 34 & 0.22 & 0.24 & 0.21 & 1.14 & 0.29 & 3 \\ 
\hline
perlbench & 155\,418 & 1\,870 & 1\,286 & 0.69 & 593 & 0.46 & 0.37 & 0.35 & 1.16 & 0.36 & 8 \\ 
\hline
sjeng & 13\,847 & 144 & 113 & 0.78 & 31 & 0.27 & 0.26 & 0.20 & 1.24 & 0.42 & 3 \\ 
\hline
sphinx3 & 25\,090 & 369 & 275 & 0.75 & 76 & 0.28 & 0.29 & 0.31 & 1.19 & 0.44 & 6 \\ 
\hline
\multicolumn{12}{c}{}\\[-1em]
\cline{7-12}
\multicolumn{6}{l|}{} & Mean & 0.26 & 0.25 & 1.26 & 0.47 & 6.08 \\
\cline{7-12}
\end{tabular} 
\end{footnotesize}
\end{adjustbox}
}
\end{small}
\vspace{-3mm}
\end{table}

$live_{opt}$ can be implemented in debuggers that can evaluate expressions over the current program state, such as \gdb\ and \lldb\footnote{\lldb\ is integrated within the LLVM infrastructure, so it can JIT-compile and run arbitrary code. \gdb\ can evaluate complex expressions, too.}. In fact, this version needs only to access the live state of the optimized program at the breakpoint.

\noindent
$avail$ can be integrated in a debugger using {\em invisible} breakpoints to spill a number of available values before they are overwritten. Invisible breakpoints are largely employed in source-level debuggers~\cite{Zellweger83,Wu99,Jaramillo00}. Using spilled values and the current live state, expected values for endangered user variables can be reconstructed as for $live_{opt}$. Alternatively, in a virtual machine with a JIT compiler and an integrated debugger, the runtime might recompile a function when the user inserts a breakpoint in it, artificially extending the liveness range for the available values possibly needed by \reconstruct.

\myfigure\ref{fig:debug-ratio} shows for each benchmark the global average recoverability ratio achieved by $live_{opt}$ and $avail$ on the set of affected functions $F_{end}$. We observe that $avail$ performs particularly well on all benchmarks, with a global ratio higher than $95\%$ for half of the benchmarks, and higher than $90\%$ for $10$ out of $12$ benchmarks. In the worst case (\mytt{gobmk}), we observe a global ratio slightly higher than $83\%$. Results thus suggest that \reconstruct\ can recover expected values for the vast majority of source-level endangered variables.

\begin{figure}[!ht]
\begin{center}
\vspace{2mm}
\includegraphics[width=0.8\textwidth]{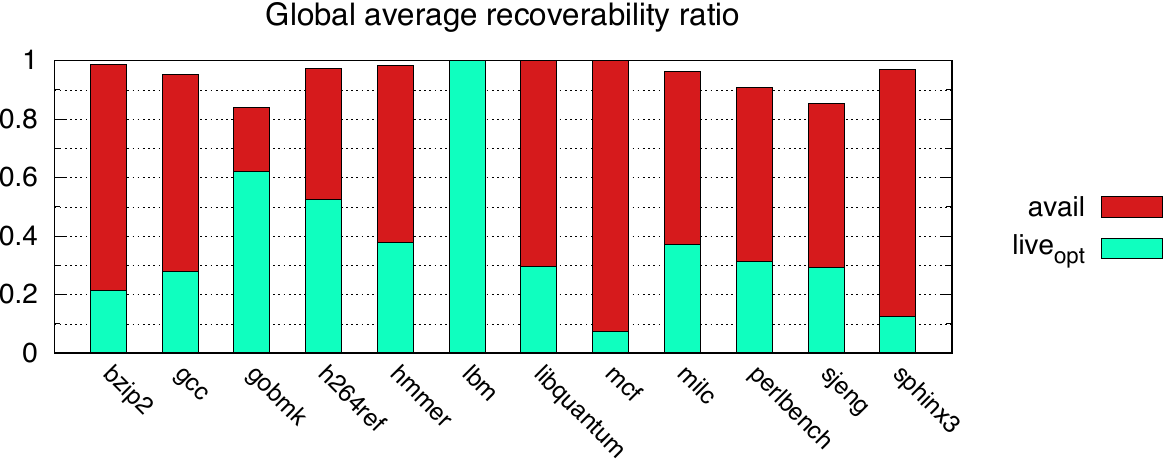}
\vspace{-2mm}
\caption{\label{fig:debug-ratio} Global average recoverability ratio, defined as the weighted average of each function's average recoverability ratio. We used the number of LLVM IR instructions in the unoptimized function code as weight.}
\end{center}
\vspace{-3mm}
\end{figure}

\begin{table}[!t]
\begin{small}
\caption{\label{tab:CS-debug-dead-avail} Available values to be preserved when using $avail$. For functions that require to preserve at least one value, we report the fraction $frac$ of $|F_{end}|$ they cumulatively account for, the average number $avg$ of values to preserve across such functions, and the corresponding standard deviation $\sigma$.}
\makebox[\textwidth][c]{
\begin{footnotesize}
\begin{tabular}{ |c|c|c|c|c|c|>{\centering}p{0.46cm}|c|c|c|c|c|c||c| }
\cline{2-14}
\multicolumn{1}{c|}{} & \rot{bzip2} & \rot{gcc} & \rot{gobmk} & \rot{h264ref} & \rot{hmmer} & \rot{lbm} & \rot{libquantum\hspace{0.5em}} & \rot{mcf} & \rot{milc} & \rot{perlbench} & \rot{sjeng} & \rot{sphinx3} & \rot{Mean} \\
\hline
$frac$ & 0.71 & 0.72 & 0.16 & 0.71 & 0.70 & - & 0.67 & 1.00 & 0.76 & 0.66 & 0.77 & 0.72 & 0.69 \\
\hline
$avg$ & 3.24 & 2.77 & 2.31 & 4.90 & 2.79 & - & 3.00 & 1.82 & 2.19 & 4.76 & 1.88 & 2.31 & 2.91 \\ 
\hline
\hline
$\sigma$ & 3.38 & 5.12 & 2.22 & 9.23 & 2.33 & - & 3.46 & 0.87 & 1.94 & 4.94 & 1.12 & 2.08 & 3.34 \\
\hline
\end{tabular} 
\end{footnotesize}
}
\end{small}
\end{table}

To estimate how many values should be preserved - through either invisible breakpoints or recompilation - to integrate $avail$ in a debugger, we collected for each function the ``keep'' set of non-live available values to save to support deoptimization across all program points corresponding to source-level locations. We then compute the average and the standard deviation for the size of this set on all the endangered functions. Figures reported in \mytable\ref{tab:CS-debug-dead-avail} show that typically a third of the endangered functions do not require any value to be preserved. For the remaining functions, $2.91$ values need to be preserved on average, with a peak of $4.90$ for \mytt{h264ref}.

Observe that values in the keep set do not necessarily need to be preserved all simultaneously or at all points: indeed, the minimal set to be maintained can change across function regions. Typically when debugging, values are saved using an invisible breakpoint before they are overwritten, and deleted as soon as they are no longer needed~\cite{Jaramillo00}. For the recompilation-based approach, the numbers reported in \mytable\ref{tab:CS-debug-dead-avail} should be interpreted in terms of possible register pressure increase as discussed in \mysection\ref{ss:experiments-discussion}.

\section{Related Work}
\label{se:related}


\paragraph{On-Stack Replacement} OSR has been pioneered in the implementations of the SELF language runtime to support dynamic deoptimization for debugging purposes~\cite{Holzle92}. The rise of the Java language has then brought OSR technology to the mass market, employing it in the most sophisticated runtimes.

In HotSpot Server~\cite{Paleczny01} OSR is employed to optimize performance-critical methods by instrumenting their entry point and backward branches, while for deoptimization execution is transferred to the interpreter when class loading invalidates an optimization decision. \cite{Fink03} describes an OSR mechanism for Jikes RVM that places instrumentation as in HotSpot to support a profile-driven deferred compilation mechanism. Jikes RVM employs OSR also to recover from speculative inlining decisions, using an OSR stub to divert execution to a newly generated function. Its compiler can generate an \mytt{OSRBarrier} instruction to capture the JVM-level program state before executing a bytecode instruction in an interruptible method.

Tracing JIT compilers insert guards at points of possible divergence for the recorded control flow. RPython~\cite{Schneider12} uses trampolines to analyze resume information for a guard and runs a compensation code to leave the trace. SPUR~\cite{Bebenita10} relies on a transfer-tail JIT to bridge the execution to the baseline JIT.

The Graal compiler~\cite{Wurthinger13} uses partial evaluation to generate aggressively optimized code, falling back to an interpreter for deoptimization. Interpreter stack frames are restored using the metadata associated with the deoptimization point, while grouping mechanisms are used to reduce the size of metadata to be globally maintained~\cite{Duboscq14} in a similar manner as in RPython and HotSpot.

\noindent
The V8 JavaScript engine implements a multi-tier compilation system with the recent addition of an interpreter. To capture modifications to the program state, the IR graph is processed in an abstract interpretation fashion, tracking changes incrementally performed by single instructions. During the lowering phase this information is then materialized as deoptimization data where needed. V8's highly optimizing TurboFan compiler supports OSR at loop headers, generating a continuation function specialized for the current variable values at the loop entry.

\paragraph{Correctness of Compiler Optimizations}
Translation validation~\cite{Pnueli98,Necula00} tackles the problem of verifying that the optimized version of a specific input program is semantically equivalent to the original program. \cite{Lacey02,Lacey04} propose to express optimizations as rewrite rules with CTL formulas as side conditions, showing how to prove such transformations correct. \cite{Lerner03,Lerner05} investigate how to automatically prove soundness for optimizations expressed as transformation rules. \cite{Kundu09} makes a further step towards generality by proving the equivalence of {\em parameterized programs}, which yields correctness of transformation rules once for all. We believe that this approach deserves further investigation in the OSR context, as it could provide a principled approach to computing mappings between equivalent points in different program versions in the presence of complex optimizations. \cite{Lopes15} presents Alive, a domain-specific language for writing provably correct LLVM peephole optimizations. Alive found several bugs in existing LLVM transformations. We look forward to future extensions that would support control flow branches in Alive.

While all the aforementioned works focus on proving optimizations sound, in this article we aim at proving OSR correct in the presence of optimizations. Of a different flavor, but in a similar spirit as ours, \cite{Guo11} uses bisimulation to study what optimizations of a tracing JIT compiler are sound. OSR is used in traditional JIT compilation to devise efficient code for a whole method, while a tracing JIT performs aggressive optimizations on a linear sequence of instructions, which control flow can leave through guarded side exits only.

\paragraph{Optimized Code Debugging} 
We now discuss the connections of the ideas presented in our case study with previous works in the debugging literature.
We are aware of only one work that supports full source-level debugging with expected behavior. TARDIS~\cite{Barr14} is a time-traveling debugger for managed runtimes that takes snapshots of the program state at a regular basis, and lets the unoptimized code run after a snapshot has been restored to answer queries. Our solution is different in the spirit, as we tackle the problem from the performance-preserving end of the spectrum~\cite{Adl-Tabatabai96thesis}, and in some ways more general, as it can be applied to the debugging of statically compiled languages such as C.

\cite{Wu99} proposes a framework to selectively take control of the execution by inserting four kinds of breakpoints, and perform a forward recovery process in an emulator that executes the optimized instructions mimicking their ordering at the source level. The emulation scheme however cannot report values whose reportability is path-sensitive. FULLDOC~\cite{Jaramillo00} makes a step further, as it can provide truthful behavior for deleted values, and expected behavior for the other values. The authors remark that FULLDOC can be integrated with techniques for reconstructing deleted values, and \reconstruct\ might be an ideal candidate.

\cite{Hennessy82} presents algorithms for recovering values in locally optimized code -- with weaker extensions to global optimizations -- that can only work with operand values that are user variables coming from memory, as they ignore compiler temporaries or registers. Unfortunately, the advances in compiler and debugging technology make a revision of the assumptions behind them necessary~\cite{Copperman93}.


\cite{Adl-Tabatabai96thesis} presents novel algorithms for value recovery in optimized programs. In particular, the algorithms for global optimizations identify compiler temporaries introduced by optimizations that alias endangered source variables. This idea is captured by our technique, which can also use facts recorded during IR manipulation (\mysection\ref{ss:buildcomp-implementation}) when recursively reconstructing portions of the original program's state.


\paragraph{Other Related Work} \cite{Bhandari15} discusses loop tiling in the presence of exception-throwing statements that thwart optimization. To roll back out-of-order updates during deoptimization, their algorithm identifies a minimal number of elements to back up and generates the necessary code. 
Product programs~\cite{Barthe11} are used to verify relational (e.g., transformations) and $k$-safety (e.g., continuity) properties; they are orthogonal to multi-version programs, which embody the notion of OSR and rely on CTL and model checking. 

\oldrevision{
\paragraph{Other Related Work}
Program slicing techniques~\cite{Weiser82,Weiser84,Korel88,Agrawal90} have found many diverse applications, such as program debugging, comprehension, analysis, testing, verification, and optimization.
We believe that the simple ideas behind our \buildcomp\ algorithm could be improved by taking advantage of this wealth of analysis techniques.


\cite{Bhandari15} explores deoptimization for loop tiling in the presence of exceptions. In order to roll back out-of-order updates, an algorithm identifies a minimal number of elements to back up and generates the necessary code. Intuitively, supporting deoptimization for complex loop transformations may be both space- and time- costly, but in the OSR context flexibility/performance trade-offs are still largely unexplored.
}

\section{Conclusions}
\label{se:conclusions}

In this article we make a first step towards a provably sound general framework for OSR, backed by promising results in real benchmarks. We run a number of unhindered LVE transformations, achieving bidirectional support for OSR at most program locations. Our algorithms can also be useful for variable reconstruction in source-level debuggers. We expect our techniques to be easily portable to other runtimes.

Our work is just a scratch off the surface of the fascinating problem of how to dynamically morph one program into another. As a next step, we plan to investigate automatic algorithms for other classes of transformations. Intuitively, supporting compensation code for heavy-duty ones might require a form of state logging (\mysection\ref{se:related}): flexibility/performance trade-offs are however still largely unexplored in the OSR context, and a deep understanding of them remains a compelling goal. 

We believe that the simple ideas behind our \buildcomp\ algorithm could be integrated with powerful program analysis techniques such as program slicing~\cite{Weiser82} in order to support OSR at even more points. We also plan to address situations where the OSR landing pad may not be unique, as in software pipelining. 

We hope to look at future tools deriving from the techniques presented in this article: interesting directions include exploiting the information collected for the instrumented passes to aid the \lldb\ debugger in expected-behavior recovery, and exploring OSR for switching between instrumented and uninstrumented code when using memory sanitizers that add checks at IR level~\cite{Wagner15}.

\bibliographystyle{abbrvnat}
\bibliography{biblio}


\newpage
\appendix

\section{Computation Tree Logic Operators}


In this section we provide formal definitions of CTL temporal operators in our language framework. In particular, their formalization will rely on the following definition of control flow graph:

\begin{definition}[Control Flow Graph]
\label{de:cfg}
The {\em control flow graph} (CFG) for a program $\pi=\langle I_1, I_2, \ldots, I_n \rangle$ is described by a pair $G=(V, E \subseteq V\times V)$ where:
\begin{align*}
V &= \{ I_1, I_2, \ldots, I_n \} \\
E &= \{(I_i, I_{i+1})\:|\: I_i \neq \textsf{abort} \wedge I_i \neq \textsf{goto m}, \!\textsf{ m}\in Num \} \\
&\cup\;\{(I_i, I_m)\:|\: I_i = \textsf{goto m} \vee I_i = \textsf{if (e) goto m}, \!\textsf{ m}\in Num, \!\textsf{ e}\in Expr \}.
\end{align*}
\end{definition}

\noindent We also need to formalize the concept of finite maximal paths:

\begin{definition}[Set of Complete Paths] Given a control flow graph $G=(V,E)$ and an initial node $n_0\in V$, the {\em set of complete paths} $CPaths(n_0,G)$ starting at $n_0$ consists of all finite sequences $\langle n_0,n_1,\ldots,n_k\rangle$ such that $(n_i,n_{i+1})\in E$ for all $n_i$ with $i<k$, and such that there does not exist a $n_{k+1}$ such that $(n_k,n_{k+1})\in E$.
\end{definition}

\noindent Complete paths from a specified node (i.e., instruction) are thus maximal finite sequences of connected nodes through a control flow graph from an initial point to a sink node, which in our setting is unique (unless {\tt abort} instructions are present) and corresponds to the final instruction $I_n$ of a program $\pi$ as in \mydefinition\ref{de:program}.

\smallskip
We can now define temporal operators as follows:

\begin{definition}[Temporal Operators]
Given a node $n$ in the control flow graph $G=(V,E)$ of a program $\pi$, we define the following CTL {\em temporal operators}:
\begin{align*}
n \models \overrightarrow{AX}(\phi) &\Longleftrightarrow \forall m: (n,m)\in E: \pi,m\models\phi \\
n \models \overrightarrow{EX}(\phi) &\Longleftrightarrow \exists m: (n,m)\in E: \pi, m\models\phi \\
n \models \overrightarrow{A}(\phi~U~\psi) &\Longleftrightarrow \forall p: p\in CPaths(n,G): Until(\pi, p,\phi,\psi) \\
n \models \overrightarrow{E}(\phi~U~\psi) &\Longleftrightarrow \exists p: p\in CPaths(n,G): Until(\pi, p,\phi,\psi)
\end{align*}
\noindent where predicate $Until(\pi,p,\phi,\psi)$ holds for $p = \langle n_0,n_1,\ldots,n_k\rangle \in CPaths(n_0,G)$ if:
\begin{equation*}
\exists j: 0 \le j\le k: \pi, n_j \models \psi \; \wedge \: \forall 0 \le i < j: \pi, n_i \models \phi
\end{equation*}
\noindent Operators $\overleftarrow{AX}$, $\overleftarrow{EX}$, $\overleftarrow{A}$, and $\overleftarrow{E}$ can be defined similarly on the reverse control flow graph $\overleftarrow{G}$, which is identical to $G$ but with every edge in $\overleftarrow{E}$ flipped.
\end{definition}

\begin{example}
Dominance analysis is widely employed in a number of program analyses and optimizations. In a CFG, we say that a node $n$ dominates a node $m$ if every path from the CFG's entry node to $m$ must go through $n$. Using CTL operators, we can easily encode this property. Given a program $\pi$ as in \mydefinition\ref{de:program}, we can write: 
\begin{equation*}
 \mytt{dominates}(n,m) \Longleftrightarrow \pi,I_1 \models \neg E(\neg\wpoint(n)~U~\wpoint(m))
\end{equation*}

\noindent 
which captures the idea that there is no path from $\pi$'s first instruction that reaches $m$ without reaching $n$ first.
\end{example}


\section{Proofs of Theorems}
In this section we provide proofs for the theorems stated in the article and present a number of related lemmas and corollaries. Multi-version programs are addressed separately in Appendix~\ref{apx:multiver}.

\subsection{OSR Mappings}

\onlylivecount*

\begin{proof}
We reason on the structure of the transition relation $\Rightarrow_{\pi}$ for our big-step semantics shown in \mydefinition\ref{de:transitions}. We rewrite our claim as:
\begin{align*}
(\sigma,l)\Rightarrow_{\pi}(\sigma',l') ~~ \Longleftrightarrow ~~ & (\sigma\vert_{\live(\pi,l)},l)\Rightarrow_{\pi}(\hat{\sigma},l') ~~
\wedge ~~ \hat{\sigma}\vert_{\live(\pi,l')} = \sigma'\vert_{\live(\pi,l')} 
\end{align*}

\noindent When \myequation(\ref{eq:asgn-sem}) applies, both states advance to location $l+1$, and the evaluation $(\sigma, \texttt{e}) \Downarrow v$ for the assignment yields the same result in both stores, as each operand in $e$ is either a constant literal or a live variable for $\pi$ at $l$. Indeed, having a variable operand for $e$ not in $\live(\pi,l)$ would contradict the definition of liveness. When the instruction at $l$ is a conditional expression, $\Rightarrow_{\pi}$ applies either \myequation(\ref{eq:ifz-sem}) or \myequation(\ref{eq:ifnz-sem}) to both states: as discussed for assignments, the evaluation of expression $e$ yields the same result in $\sigma$ and $\sigma\vert_{\live(\pi,l)}$, and both states advance to the same location without affecting the store. When one of \Crefrange{eq:goto-sem}{eq:out-sem} applies, trivially both states advance to the same location, while values in their stores are not affected. Finally, from \mydefinition\ref{de:live-var} it follows that $\live(\pi,l')\supseteq\live(\pi,l)\cup\{\,\texttt{x} ~ | ~ I_l=\texttt{x:=e}\,\}$ and thus $\hat{\sigma}\vert_{\live(\pi,l')} = \sigma'\vert_{\live(\pi,l')}$.
\end{proof}

\subsection{LVE Transformations and OSR Mapping Generation Algorithms}

\begin{restatable}{lem}{bisimprop}
\label{le:bisim-prop}
Let $R$ be a reflexive bisimulation relation between programs $\pi$ and $\pi'$. Then for any $\sigma\in \Sigma$ it holds:
\begin{equation}
\label{eq:bisim-prop-1}
|\tau_{\pi\sigma}|=|\tau_{\pi'\sigma}|
\end{equation}
\begin{equation}
\label{eq:bisim-prop-2}
\forall i\in dom(\tau_{\pi\sigma}), ~~ \tau_{\pi\sigma}[i]~R~\tau_{\pi'\sigma}[i]
\end{equation}
\end{restatable}

\begin{proof}
We prove \myequation(\ref{eq:bisim-prop-2}) by induction on $i$. The base follows from $\tau_{\pi\sigma}[0]=\tau_{\pi'\sigma}[0]=(\sigma,1)$ and the assumption that $R$ is reflexive. Assume as an inductive hypothesis that $\tau_{\pi\sigma}[i]~R~\tau_{\pi'\sigma}[i]$ for any $i<|\tau_{\pi\sigma}|$. Since $|\tau_{\pi\sigma}|>i$ then $\tau_{\pi\sigma}[i] \Rightarrow_{\pi} \tau_{\pi\sigma}[i+1]$ by \mydefinition\ref{de:exec-trace}. It follows by \mydefinition\ref{de:bisimulation} that $\tau_{\pi\sigma}[i+1]~R~\tau_{\pi'\sigma}[i+1]$.

To prove \myequation(\ref{eq:bisim-prop-1}), assume by contradiction that $|\tau_{\pi\sigma}|\neq|\tau_{\pi'\sigma}|$, e.g., $|\tau_{\pi\sigma}|>|\tau_{\pi'\sigma}|=k$. Since $|\tau_{\pi\sigma}|>k$ then $\tau_{\pi\sigma}[k]~R~\tau_{\pi'\sigma}[k]$ by \myequation(\ref{eq:bisim-prop-2}) and $\tau_{\pi\sigma}[k] \Rightarrow_{\pi} \tau_{\pi\sigma}[k+1]$ by \mydefinition\ref{de:exec-trace}. It follows by \mydefinition\ref{de:bisimulation} that $\tau_{\pi'\sigma}[k] \Rightarrow_{\pi'} \tau_{\pi'\sigma}[k+1]$. Hence $|\tau_{\pi'\sigma}|>k$, contradicting the initial assumption. The proof for the case $|\tau_{\pi'\sigma}|>|\tau_{\pi\sigma}|$ is analogous.
\end{proof}


\noindent One consequence of \mydefinition\ref{de:state-equiv-relation}, which simplifies our formal discussion, is the following:

\begin{restatable}{lem}{lvbsameloc}
\label{le:lvb-same-loc}
If $\pi$ and $\pi'$ are live-variable bisimilar, then for any $\sigma$, corresponding states in program traces $\tau_{\pi\sigma}$ and $\tau_{\pi'\sigma}$ are located at the same program points: $\forall i:$ $\tau_{\pi\sigma}[i]=(\sigma_i, l_i)$ $\wedge$ $\tau_{\pi'\sigma}[i]=(\sigma'_i, l'_i)$ $\Longrightarrow$ $l_i=l'_i$.
\end{restatable}

\begin{proof}
Straightforward by \mylemma\ref{le:bisim-prop} and \mydefinition\ref{de:state-equiv-relation}.
\end{proof}

\begin{restatable}{cor}{samesize}
\label{co:same-size}
If $\pi$ and $\pi'$ are live-variable bisimilar, then they have the same size: $\pi=\langle I_1,\ldots,I_n\rangle$ $\wedge$ $\pi'=\langle I'_1,\ldots,I'_{n'}\rangle$ $\Longrightarrow$ $n=n'$.
\end{restatable}

\begin{proof}
By \mylemmas\ref{le:bisim-prop} and~\ref{le:lvb-same-loc} and \myequation(\ref{eq:out-sem}), for any finite trace $\tau_{\pi\sigma}$ it holds $\tau_{\pi\sigma}[|\tau_{\pi\sigma}|]=(-,n+1)$ and $\tau_{\pi'\sigma}[|\tau_{\pi'\sigma}|]=(-,n+1)$. Hence both $\pi$ and $\pi'$ contain $n$ instructions.
\end{proof}

\noindent We finally introduce one more, fundamental lemma required to prove \mytheorem\ref{th:osr-trans-correctness} correct:

\begin{restatable}[Correctness of Algorithm \buildcomp]{lem}{buildcompcorr}
\label{le:build-comp-corr}
Let $\pi$ and $\pi'$ be live-variable bisimilar programs. For each initial store $\sigma\in \Sigma$ it holds:
\begin{equation*}
\forall i\in dom(\tau_{\pi\sigma}): \chi\neq undef \implies
\mysem{\chi}(\sigma_i\vert_{\live(\pi,l_i)}) = \sigma'_i\vert_{\live(\pi',l_i)}
\end{equation*}
where $(\sigma_i,l_i)=\tau_{\pi\sigma}[i]$, $(\sigma'_i,l_i)=\tau_{\pi'\sigma}[i]$, and $\chi={\tt build\_comp}(\pi,l_i,\pi',l_i)$.
\end{restatable}

\begin{proof}
The correctness of ${\tt build\_comp}$ (\myalgorithm\ref{alg:osr-build-comp}) relies on the ability of \reconstruct\ (\myalgorithm\ref{alg:osr-reconstruct}) to produce compensation code for each variable that is live at the OSR destination, but not at the origin.
Procedure $\reconstruct(\wx, \pi, l, \pi', l', l'')$ aims at creating a sequence of instructions that assigns \wx\ with the value that it would have assumed at $l''$ in $\pi'$, using as input the values of the live variables at $l$ in $\pi$.

We proceed by induction on the recursive calls of \reconstruct. For the algorithm to succeed, there must be a unique definition {\tt x:=e} at some point $\hat{l}$ that dominates $l''$, otherwise $undef$ is thrown (see \myfigure\ref{fig:reconstruct}). 

\begin{figure}[!ht]
\begin{center}
\includegraphics[width=0.42\textwidth]{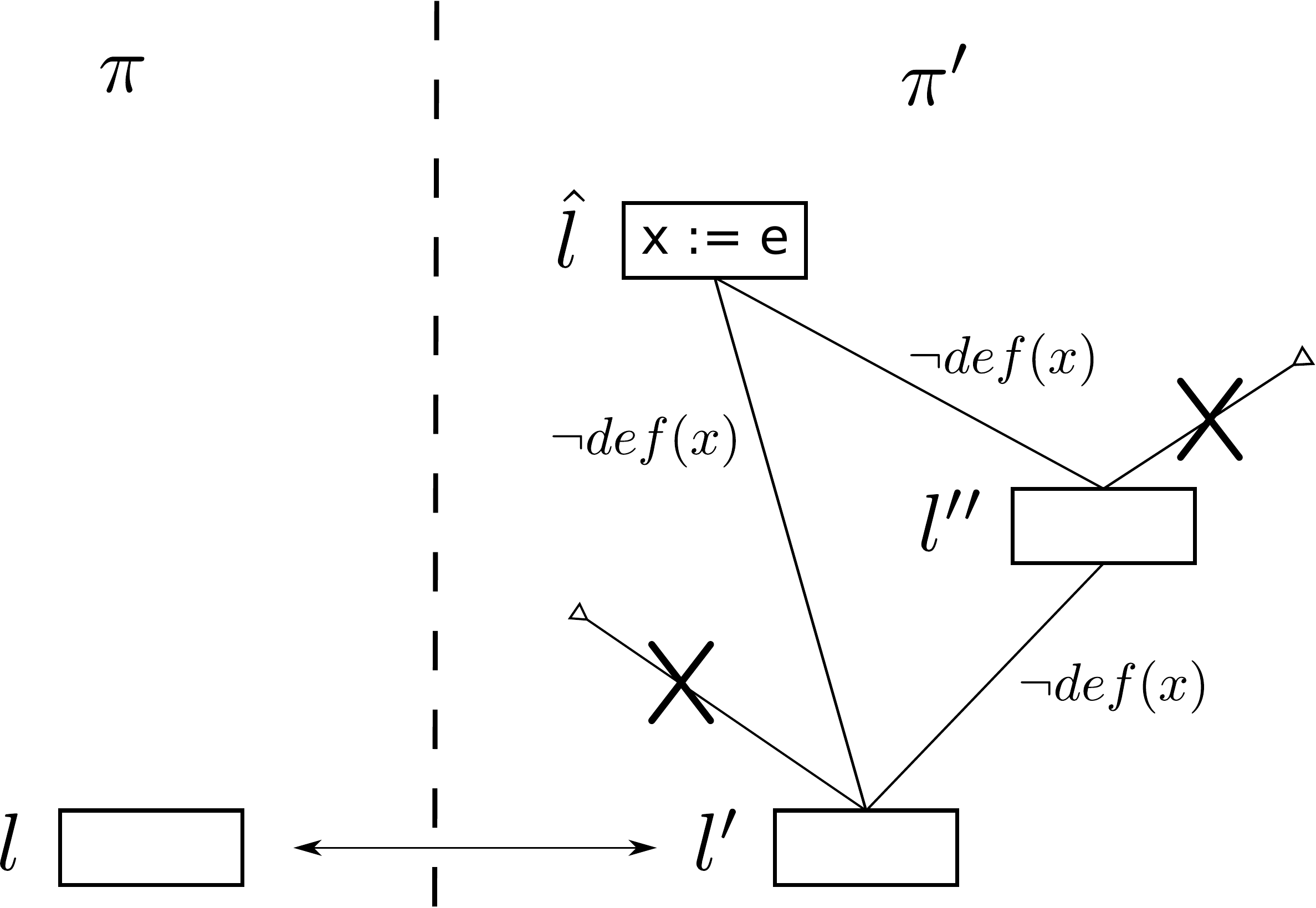}
\caption{\label{fig:reconstruct} \reconstruct\ identifies an assignment \mytt{x:=e} at $\hat{l}$ that reaches both $l'$ and $l''$, and no other definition of \wx\ is possible.}
\end{center}
\end{figure}

\noindent The base case happens when either:
\begin{enumerate}[itemsep=3pt, parsep=0pt]
 \item {\tt e} has no free variables (line~\ref{line-rec:freevar}), thus the compensation code for \wx\ is just {\tt x:=e} (line~\ref{line-rec:assign});
 \item the definition at $\hat{l}$ reaches both $l''$ and $l'$ (lines~\ref{line-rec:first-inst},~\ref{line-rec:both-live}) and \wx\ is live at both the origin and the destination (line~\ref{line-rec:both-live}), hence, since $\pi$ and $\pi'$ are live-variable bisimilar and \wx\ has the same value at $l$ and $l'$, no compensation code for \wx\ is needed as the value of \wx\ at $l$ is the same that it would have had at $l''$;
 \item $\hat{l}$ has already been visited, so compensation code for \wx\ has already been created.
\end{enumerate}

\noindent Assume by inductive hypothesis that the recursive calls of \reconstruct\ have added to $\chi$ the code to assign each free variable $y$ of $e$ with the value they would have assumed at $\hat{l}$ (line~\ref{line-rec:recursive}). Then the value of \wx\ that we would have had at $\hat{l}$ is determined by \mytt{x:=e}, which is appended to $\chi$ (line~\ref{line-rec:assign}).
\end{proof}

\osrtranscorrectness*

\begin{proof}
The correctness of ${\tt OSR\_trans}$ follows directly by \mylemma\ref{le:lvb-same-loc}, \mylemma\ref{le:build-comp-corr}, and \mycorollary\ref{co:same-size}.
\end{proof}

\lvetransexamples*

\begin{proof}
CP replaces uses of a variable $v$ at a node $m$ with a constant $c$ when all the reaching definitions for v are of the form $v:=c$. DCE deletes an instruction at a node $m$ if the result of its computation will never be used later in the execution, skipping past possible uses of the $x$ itself at $m$ with AX. Hoist moves an assignment of the form $x:=v[e]$ from a node $q$ to an insertion point $p$ provided that two conditions are met: (1) in all forward paths starting at $p$, $x$ is not used until the original location $q$ is reached; and (2) in all backward paths starting at $q$, $x$ is not reassigned at any node other than $q$ and the constituents of $e$ are not redefined, until $p$ is reached.

\noindent In \cite{Lacey02}, CP, DCE, and Hoist are proved correct, each using a different bisimulation relation $R$. 
For CP, $R$ is simply the identity relation, hence $A(l)=Val\supseteq \live(\pi,l)\cap \live(\pi',l)$ in \mydefinition\ref{de:state-equiv-relation}.

For the other two transformations, $R$ is piecewise-defined on the indexes of the traces. For any initial store $\sigma\in \Sigma$, let $\tau_{\pi\sigma}[i]=(\sigma_i,l_i)$, $\tau_{\pi'\sigma}[i]=(\sigma'_i,l'_i)$, and $t$ be the index of the final state in both traces (note that $|\tau_{\pi\sigma}|=|\tau_{\pi'\sigma}|$ from \mylemma\ref{le:bisim-prop}). Let also $\theta$ be a substitution that bounds free meta-variables with concrete program objects so that a rule's side-condition is satisfied.

\noindent
For DCE, $R$ is the identity relation before the eliminated assignment \mytt{x:=e}, and $A(l)=Val\setminus\{\theta(\wx)\}=\live(\pi,l)\cap \live(\pi',l)$ after it. $R$ is a bisimulation such that $\forall i\in[1,t]$ $l_i=l'_i$ and both the following conditions hold:
\begin{enumerate}
 \item $[\forall j, j<i \Rightarrow l_j\neq \theta(p)] \Rightarrow \sigma_i=\sigma'_i~$ and
 \item $[\exists j, j\leq i \wedge l_j = \theta(p)] \Rightarrow \sigma_i\setminus\texttt{x}=\sigma'_i\setminus\texttt{x}$
\end{enumerate}
\noindent where $p$ is the meta-variable for the eliminated assignment in $\pi'$, and $\sigma\setminus\texttt{x}$ is syntactic sugar for $\sigma\vert_{D(\sigma)}$, where $D(\sigma)=\{v\in Var ~|~ v\neq \texttt{x} ~\wedge~ \sigma(v)\neq\bot\}$ is the set of all the variable identifiers other than \mytt{x} currently defined in $\sigma$.

For Hoist, $R$ is the identity relation before $\theta(p)$ and after $\theta(q)$ (see \myfigure\ref{fig:sample-trans}), and $A(l)=Val\setminus\{\theta(\wx)\}=\live(\pi,l)\cap \live(\pi',l)$ between them. Formally, we have that $\forall i\in[1,t]$ $l_i=l'_i$ and one of the following cases holds:
\begin{enumerate}
 \item $\sigma_t=\sigma'_t ~\wedge~ \forall i~[0\leq i<t ~\Rightarrow~ l_i \notin \{\theta(p),~\theta(q)\}]$
 \item $\begin{aligned}[t]
\sigma_t=\sigma'_t ~\wedge~ \exists i~[&0\leq i<t ~\wedge~ l_i=\theta(q) ~\wedge~ \sigma_i=\sigma'_i ~\wedge\\
&\forall j~(i<j<t ~\Rightarrow~ l_j \notin \{\theta(p),~\theta(q)\})]
       \end{aligned}$
 \item $\begin{aligned}[t] 
\exists i~[0\leq i<t ~\wedge~ &l_i=\theta(p) ~\wedge~ (\sigma_t\setminus\texttt{x}=\sigma_t'\setminus\texttt{x}) ~\wedge~ (\sigma_i\setminus\texttt{x}=\sigma_i'\setminus\texttt{x}) ~\wedge\\
&\forall j~(i<j<t ~\Rightarrow~ l_j \notin \{\theta(p),~\theta(q)\}]
       \end{aligned}$
\end{enumerate}

\noindent Case 1 applies before $\theta(p)$ is reached in the trace. Case 3 applies after $\theta(p)$ has been reached, but $\theta(q)$ has not. Finally, case 2 applies after $\theta(q)$ has been reached.
\end{proof}

\subsection{OSR Mapping Composition}

\begin{restatable}[Semantics of program composition]{lem}{progcompsem}
\label{le:prog-comp-sem}
Let $\pi,\pi'\in Prog$ be any pair of composable programs, then $\forall\sigma\in\Sigma,$ $\mysem{\pi\circ\pi'}(\sigma)=\mysem{\pi'}\left(\mysem{\pi}(\sigma)\right)$.
\end{restatable}

\begin{proof}
Straightforward by \mydefinitions\ref{de:program-semantics} and~\ref{de:composition}.
\end{proof}

\osrmappingcomp*

\begin{proof} 
Let $\mu_{\pi\pi''}=\mu_{\pi\pi'}\circ\mu_{\pi'\pi''}$. By \mydefinition\ref{de:osr-mapping}, it holds:
\begin{small}
\begin{gather*}
\forall \sigma\in\Sigma, \forall s_i=(\sigma_i,l_i)\in\tau_{\pi\sigma}: ~l_i\in dom(\mu_{\pi\pi''}),\\
\exists \sigma',\sigma''\in\Sigma,~\exists s_j=(\sigma_j,l_j)\in\tau_{\pi'\sigma'}, \\ \exists s_k=(\sigma_k,l_k)\in\tau_{\pi''\sigma''}: ~
\mu_{\pi\pi''}(l_i)=(l_k,\chi\circ\chi')~\wedge~\\ \mysem{\chi\circ\chi'}(\sigma_i\vert_{\live(\pi,l_i)})=\text{[by \mylemma\ref{le:prog-comp-sem}]}\\
\mysem{\chi'}(\mysem{\chi}(\sigma_i\vert_{\live(\pi,l_i)}))=
\mysem{\chi'}(\sigma_j\vert_{\live(\pi',l_j)})=\sigma_k\vert_{\live(\pi'',l_k)}
\end{gather*}
\end{small}
\noindent Hence, $\mu_{\pi\pi'}\circ\mu_{\pi'\pi''}$ is an OSR mapping from $\pi$ to $\pi''$.
\end{proof}

\begin{restatable}{cor}{composestrict}
\label{co:compose-strict}
Let $\pi,\pi',\pi''\in Prog$, let $\mu_{\pi\pi'}$ and $\mu_{\pi'\pi''}$ be strict OSR mappings as in \mydefinition\ref{de:osr-mapping}. Then $\mu_{\pi\pi'}\circ\mu_{\pi'\pi''}$ is a strict OSR mapping from $\pi$ to $\pi''$.
\end{restatable}

\begin{proof}
Straightforward by \mydefinition\ref{de:osr-mapping} and \mytheorem\ref{thm:osr-mapping-comp}.
\end{proof}


\section{Multi-version Programs}
\label{apx:multiver}

In this section we discuss multi-version programs in detail, providing the machinery required to prove \mytheorem{\ref{th:mv-prog-determ}} correct, and describe a multi-pass transformation algorithm for constructing multi-version programs.

To characterize the execution behavior of a multi-version program, we consider the system of traces of an execution transition system that start from a given initial state.

\begin{definition}[Trace System of Multi-Version Program]
\label{de:mvp-exec-system}
The system of traces ${\mathcal T}_{\Pi,\sigma}$ contains all traces $\tau$ of transition system $(MState,\Rightarrow_{\Pi})$ such that $\tau[0]=(1,\sigma,1)$.
\end{definition}

\begin{definition}[Deterministic Multi-Version Program]
\label{de:deterministic-mvp}
A multi-version program $\Pi$ is deterministic iff $\forall \sigma\in\Sigma$, either all traces in ${\mathcal T}_{\Pi,\sigma}$ are infinite, or they all lead to the same store, i.e.:
\begin{gather*}
\forall \tau, \tau'\in{\mathcal T}_{\Pi,\sigma}: ~~ \big(|\tau|=\infty ~ \Longleftrightarrow ~ |\tau'|=\infty\big) ~ \wedge \\
\big(|\tau|<\infty ~ \Longrightarrow ~ \exists~p,p',l,l'\in\mathbb{N}, \sigma,\sigma'\in\Sigma: ~
\tau[|\tau|]=(p,\sigma,l) ~ \wedge ~ \tau'[|\tau'|]=(p',\sigma',l') ~ \wedge ~ \sigma=\sigma' \big)
\end{gather*}
\end{definition}

\noindent The meaning of a deterministic multi-version program can be defined as follows:

\begin{definition}[Multi-Version Semantic Function]
\label{de:mv-program-semantics}
The semantic function $\mysem{\Pi}:\Sigma \rightarrow \Sigma$ of a deterministic multi-version program $\Pi$ is defined as: 
$$
\forall \sigma\in\Sigma: ~~ \mysem{\Pi}(\sigma)=\sigma' ~~ \Longleftrightarrow ~~ (1,\sigma,1) \Rightarrow^{*}_{\Pi} (p,\sigma',|\pi_p|+1)
$$
where $\Rightarrow^{*}_{\Pi}$ is the transitive closure of $\Rightarrow_{\Pi}$.
\end{definition}

\noindent
To prove the correctness of this approach, we introduce a preliminary lemma and then use it to prove that a multi-version program built in this way is deterministic. 

\begin{restatable}{lem}{complemma}
\label{le:comp-lemma}
Let $\tau\in{\mathcal T}_{\Pi,\sigma}$ be an execution trace in the system of the traces for the multi-version program $\Pi$ $=({\mathcal V}, {\mathcal E}, {\mathcal M})$ constructed using \dopasses\ and LVE transformations, and let $\omega_1,\ldots,\omega_k$ be the indexes of $\tau$ where an OSR transition has just occurred, with $\tau[\omega_i]=(p_{\omega_i}, \sigma_{\omega_i}, l_{\omega_i})$. Then $\forall i\in[1,k]$ there exists a state $(\hat{\sigma}_i,\hat{l}_i)$ in the trace of $\pi_{p_{\omega_{i}}}$ starting from the initial store $\sigma$ such that $\hat{l}_i=l_{\omega_i}$ and $\hat{\sigma}_i\vert_{\live(\pi_{p_{\omega_{i}}},\,\hat{l}_i)}=\sigma_{\omega_i}\vert_{\live(\pi_{p_{\omega_{i}}},\,\hat{l}_i)}$.
\end{restatable}

\begin{proof}
To simplify the notation we introduce:
\begin{equation*}
\hat{\pi}_i = \begin{cases}
\pi_1 & \text{if } i=0\\
\pi_{p_{\omega_i}} & \text{if } i \in [1,k]
\end{cases}
\end{equation*}

\noindent From \myequation(\ref{eq:mv-big-step}) we can write that $\tau[\omega_i]=(p_{\omega_i}, \sigma_{\omega_i},l_{\omega_i})$ has been obtained from $\tau[\omega_i-1]=(p_{\omega_i-1}, \sigma_{\omega_i-1}, l_{\omega_i-1})$ with $\sigma_{\omega_i}=\mysem{\chi_{\omega_i-1}}(\sigma_{\omega_i})$. For each OSR transition $\hat{\pi}_i$ has been obtained from $\hat{\pi}_{i-1}$ using \dopasses\ for some sequence $L$ of LVE transformations. Indeed, in order for \myequation(\ref{eq:mv-big-step}) to apply:
\begin{equation*}
(\hat{\pi}_{i-1},\hat{\pi}_i)\in{\mathcal E} ~\wedge~ \exists L: ~{\tt do\_passes}(\hat{\pi}_{i-1}, L-1)=(\hat{\pi}_i,\mu_{\hat{\pi}_{i-1}\hat{\pi}_i},{\mu'}_{\hat{\pi}_{i}\hat{\pi}_{i-1}}) ~\wedge~
M(\hat{\pi}_{i-1},\hat{\pi}_i)=\mu_{\hat{\pi}_{i-1}\hat{\pi}_i}
\end{equation*}

\noindent When the OSR step is performed we thus have:
\begin{equation*}
M(\hat{\pi}_{i-1},\hat{\pi}_i)(l_{\omega_i-1})=\mu_{\hat{\pi}_{i-1}\hat{\pi}_i}(l_{\omega_i-1})=(l_{\omega_i},\chi_{\omega_i-1})
\end{equation*}

\noindent By \mytheorem\ref{th:osr-trans-correctness} function $\mu_{\hat{\pi}_{i-1}\hat{\pi}_i}$ provides a strict OSR mapping between $\hat{\pi}_{i-1}$ and $\hat{\pi}_i$, as all LVE transformations in L are composed into a strict mapping (\mycorollary\ref{co:compose-strict}). Note also that since $\Delta_I$ is being used to map OSR program points between $\hat{\pi}_{i-1}$ and $\hat{\pi}_i$, it follows that $l_{\omega_i}=l_{\omega_i-1}~\forall i\in[1,k]$.
We now prove our claim by induction on $i$.

\paragraph{Base step} When $i=1$, we know that no OSR transition has been performed till $l_{\omega_1-1}$ and $\hat{\pi}_0$ has been executing all the time. Then we can write:
\begin{equation*}
(1, \sigma, 1)\trans^*_{\Pi}(1,\sigma_{\omega_1-1},l_{\omega_1-1}) \Longleftrightarrow (\sigma, 1) \trans^*_{\hat{\pi}_0} (\sigma_{\omega_1-1},l_{\omega_1-1})
\end{equation*}

\noindent Trivially, $(\sigma_{\omega_1-1},l_{\omega_1-1})\in\tau_{\hat{\pi}_0\sigma}$. We can thus infer from \mydefinition\ref{de:osr-mapping}:
\begin{equation*}
\exists s_j=(\sigma_j,l_j)\in\tau_{\hat{\pi}_1\sigma}: ~\mu_{\hat{\pi}_{0}\hat{\pi}_1}(l_{\omega_1-1})=(l_j,\chi)~\wedge~
\mysem{\chi}(\sigma_{\omega_1-1}\vert_{\live(\hat{\pi}_{0},\,l_{\omega_1-1})})=\sigma_j\vert_{\live(\hat{\pi}_{1},\,l_j)}
\end{equation*}

\noindent From the definition of $\mu_{\hat{\pi}_{0}\hat{\pi}_1}$ it follows that $\chi=\chi_{\omega_1-1}$ and $l_j=l_{\omega_1}=l_{\omega_1-1}$. To prove the claim we need to show that:
\begin{equation*}
\sigma_j\vert_{\live(\hat{\pi}_{1},\,l_{\omega_1})}=\sigma_{\omega_1}\vert_{\live(\hat{\pi}_{1},\,l_{\omega_1})}
\end{equation*}

\noindent which follows directly from \mylemma\ref{le:build-comp-corr} and \mytheorem\ref{thm:osr-mapping-comp}.

\paragraph{Inductive step} As an inductive hypothesis we assume that $\exists (\hat{\sigma}_{k-1},\hat{l}_{k-1})\in\tau_{\hat{\pi}_{k-1}\sigma}$ such that:
\begin{equation*}
\hat{l}_{k-1}=l_{\omega_{k-1}} ~\wedge~ \hat{\sigma}_{k-1}\vert_{\live(\hat{\pi}_{k-1},\,\hat{l}_{k-1})}=\sigma_{\omega_{k-1}}\vert_{\live(\hat{\pi}_{k-1},\,\hat{l}_{k-1})}
\end{equation*}

\noindent
Since no OSR is performed between $\tau[\omega_{k-1}]$ and $\tau[\omega_k-1]$ we can write:
\begin{equation*}
(\hat{\sigma}_{k-1}, l_{\omega_{k-1}}) \trans^*_{\hat{\pi}_{k-1}} \cdots \trans^*_{\hat{\pi}_{k-1}} (\tilde{\sigma},l_{\omega_k-1}) ~\Longleftrightarrow ~
(\sigma_{\omega_{k-1}}, l_{\omega_{k-1}}) \trans^*_{\hat{\pi}_{k-1}} \cdots \trans^*_{\hat{\pi}_{k-1}} (\sigma_{\omega_k-1},l_{\omega_k-1})
\end{equation*}

\noindent in the same number of steps, with $\tilde{\sigma}\vert_{\live(\hat{\pi}_{k-1},\,l_{\omega_k-1})} = \sigma_{\omega_k-1}\vert_{\live(\hat{\pi}_{k-1},\,l_{\omega_k-1})}$ by \mytheorem\ref{thm:only-live-count}. Since $(\tilde{\sigma},l_{\omega_k-1})\in\tau_{\hat{\pi}_{k-1}\sigma}$ by the strictness of the OSR mapping $\mu_{\hat{\pi}_{k-1}\hat{\pi}_{k}}$:
\begin{align*}
\exists s_j=(\sigma_j,l_j)\in\tau_{\hat{\pi}_k\sigma}: ~\mu_{\hat{\pi}_{k-1}\hat{\pi}_{k}}(l_{\omega_k-1})=(l_j,\chi)~\wedge \mysem{\chi}(\tilde{\sigma}\vert_{\live(\hat{\pi}_{k-1},\,l_{\omega_k-1})})=\sigma_j\vert_{\live(\hat{\pi}_{k},\,l_j)}
\end{align*}

\noindent From the definition of $\mu_{\hat{\pi}_{k-1}\hat{\pi}_k}$ it follows that $\chi=\chi_{\omega_k-1}$ and $l_j=l_{\omega_k}=l_{\omega_k-1}$. By \mylemma\ref{le:build-comp-corr} and \mytheorem\ref{thm:osr-mapping-comp} we thus prove:
\begin{align*}
\sigma_j\vert_{\live(\hat{\pi}_{k},\,l_{\omega_k})} &= \mysem{\chi_{\omega_k-1}}(\tilde{\sigma}\vert_{\live(\hat{\pi}_{k-1},\,l_{\omega_k-1})})\\
&=\mysem{\chi_{\omega_k-1}}(\sigma_{\omega_k-1}\vert_{\live(\hat{\pi}_{k-1},\,l_{\omega_k-1})})\\
&=\sigma_k\vert_{\live(\hat{\pi}_{k},\,l_{\omega_k})})
\end{align*}
\end{proof}

\paragraph{Generation Algorithm and Correctness}
A natural way to generate a multi-version program consists in starting from a base program $\pi_1$ and constructing a tree of different versions, where each version is derived from its parent by applying one or more transformations. Algorithm \dopasses\ reported in \myalgorithm\ref{alg:osr-trans-compose} takes a program $\pi$ and a list of program transformations, and applies them to $\pi$, producing a bidirectional OSR mapping $\mu_{\pi\pi''},\mu_{\pi''\pi}$ between $\pi$ and the resulting program $\pi''$.
Its correctness follows by induction from \mytheorem\ref{thm:osr-mapping-comp}. Using this approach, it is straightforward to construct a multi-version program $\Pi=({\mathcal V}, {\mathcal E}, {\mathcal M})$ such that:
\begin{align*}
(\pi_p,\pi_q)\in {\mathcal E} ~~ \Longleftrightarrow ~~ \exists L: ~ &{\tt do\_passes}(\pi_p,L)=(\pi_q,\mu,\mu') ~ \wedge ~ {\mathcal M}(\pi_p,\pi_q)=\mu ~~ \vee \\
&{\tt do\_passes}(\pi_q,L)=(\pi_p,\mu,\mu') ~ \wedge ~ {\mathcal M}(\pi_p,\pi_q)=\mu'
\end{align*}

\begin{figure}[ht!]
\IncMargin{2em}
\begin{algorithm}[H]
\DontPrintSemicolon
\LinesNumbered
\SetAlgoNoLine
\SetNlSkip{1em} 
\Indm\Indmm
\KwIn{Program $\pi$, list of program transformations $L$}
\KwOut{Program $\hat{\pi}$, mappings $\mu_{\pi\hat{\pi}}$ and $\mu_{\hat{\pi}\pi}$}
\vspace{1mm} 
\nonl$\mathbf{algorithm} \> \> \dopasses$($\pi$, $T::L$)$\rightarrow$($\pi''$, $\mu_{\pi\pi''}$, $\mu_{\pi''\pi}$):\;
\everypar={\nl}
\Indp\Indpp
$(\pi',\mu_{\pi\pi'},\mu_{\pi'\pi})\gets \osrtrans(\pi,T)$\;
\lIf{$L=Nil$}{
    \Return{$(\pi',\mu_{\pi\pi'},\mu_{\pi'\pi})$}
}
$(\pi'',\mu_{\pi'\pi''},\mu_{\pi''\pi'})\gets \dopasses(\pi',L)$\;
\Return{$(\pi'',\mu_{\pi\pi'}\circ\mu_{\pi'\pi''},\mu_{\pi''\pi'}\circ\mu_{\pi'\pi})$}\;
\Indm\Indmm
\DecMargin{2.5em}
\caption{\label{alg:osr-trans-compose} OSR-aware multi-pass program transformations.}
\IncMargin{2.5em}
\end{algorithm} 
\vspace{-3mm} 
\end{figure}

\mvprogdeterm*

\begin{proof}
To prove that $\Pi$ is deterministic, we need to show that, given any initial store $\sigma$ on which $\pi_1\in\Pi$ terminates on some final state $\sigma'=\mysem{\pi_1}(\sigma)$, any execution trace $\tau\in{\mathcal T}_{\Pi,\sigma}$ terminates with $\sigma'$.

Let $\omega_1,\ldots,\omega_k$ be the indexes of $\tau$ where an OSR transition has just occurred, i.e., for any $i\in[1,k]$, state $\tau[\omega_i]$ is obtained from $\tau[\omega_i-1]$ by applying compensation code $\chi_{\omega_i-1}$ on store $\sigma_{\omega_i-1}$, which yields a store $\sigma_{\omega_i}$. The transition leads from a point $l_{\omega_i-1}$ in version $\pi_{p_{\omega_i-1}}$ to a point $l_{\omega_i}=l_{\omega_i-1}$ in version $\pi_{p_{\omega_{i}}}$ in $\Pi$. 

\noindent By \mylemma\ref{le:comp-lemma}, $\forall i\in[1,k]$ there exists a state $(\hat{\sigma}_i,\hat{l}_i)$ in the trace of $\hat{\pi}_i=\pi_{p_{\omega_{i}}}$ starting from the initial store $\sigma$ such that $\hat{l}_i=l_{\omega_i}$ and $\hat{\sigma}_i\vert_{\live(\hat{\pi}_i,\hat{l}_i)}=\sigma_{\omega_i}\vert_{\live(\hat{\pi}_i,\hat{l}_i)}$. Hence, since no OSR is fired after $\omega_k$, by \myequation(\ref{eq:mv-big-step}) it holds:
\begin{equation*}
(\hat{\pi}_k,\sigma_{\omega_{k}},l_{\omega_{k}})\trans^*_{\Pi}(\hat{\pi}_k,\sigma',|\hat{\pi}_k|+1) \Longleftrightarrow (\sigma_{\omega_{k}},l_{\omega_{k}})\trans^*_{\hat{\pi}_k}(\sigma',|\hat{\pi}_k|+1)
\end{equation*}

\noindent We can then apply \mytheorem\ref{thm:only-live-count} and \mylemma\ref{le:comp-lemma} to write:
\begin{gather*}
(\sigma_{\omega_{k}},l_{\omega_{k}})\trans^*_{\hat{\pi}_k}(\sigma',|\hat{\pi}_k|+1) \Longleftrightarrow \\ 
(\sigma_{\omega_{k}}\vert_{\live(\hat{\pi}_k,l_{\omega_{k}})},l_{\omega_{k}})\trans^*_{\hat{\pi}_k}(\sigma',|\hat{\pi}_k|+1) \Longleftrightarrow \\
(\hat{\sigma}_k\vert_{\live(\hat{\pi}_k,\hat{l}_k)},\hat{l}_k)\trans^*_{\hat{\pi}_k}(\sigma',|\hat{\pi}_k|+1)
\end{gather*}

\noindent As $(\hat{\sigma}_k,\hat{l}_k)\in\tau_{\hat{\pi}_k\sigma}$, by \mytheorem\ref{thm:only-live-count} necessarily $\sigma'=\mysem{\hat{\pi}_k}(\sigma)$. 
Given that all programs in $\Pi$ are semantically equivalent, we can conclude that $\mysem{\Pi}(\sigma)=\sigma'=\mysem{\hat{\pi}_k}(\sigma)=\mysem{\pi_1}(\sigma)$.
\end{proof}

\section{Additional Tables and Figures}
\label{apx:additional-material}


\begin{table}[!hb]
\vspace{-1mm}
\begin{small}
\caption{\label{tab:OSR-alC-bench-desc} Optimizations and utility effective on the hottest function of each benchmark. Optimization passes have been applied in the same order (left-to-right) as they appear in the table. Utility passes {\em LC} and {\em LCSSA} are prerequisites of {\em LICM}.}
\begin{footnotesize}
\vspace{-1.5mm}
\begin{tabular}{ |c|c|c|c|c|c|c|c|c|c| }
        \cline{3-10}
        \multicolumn{2}{l|}{} & \multicolumn{6}{c|}{Optimizations} & \multicolumn{2}{c|}{Utilities} \\
        \hline
        Suite & Benchmark & \em{ADCE} & \em{CP} & \em{CSE} & \em{SCCP} & \em{LICM} & \em{Sink} & \em{LC} & \em{LCSSA} \\
        \hline
        \hline
        \multirow{7}{*}{SPEC} & bzip2 & & & \checkmark & & \checkmark & \checkmark & & \checkmark \\
        \cline{2-10}
        & h264ref & \checkmark & & \checkmark & & \checkmark & \checkmark & \checkmark & \checkmark \\
        \cline{2-10}
        & hmmer & & & \checkmark & & \checkmark & \checkmark & & \checkmark \\
        \cline{2-10}
        & namd & \checkmark & \checkmark & \checkmark & \checkmark & \checkmark & \checkmark & & \checkmark \\
        \cline{2-10}
        & perlbench & \checkmark & & \checkmark & & \checkmark & \checkmark & \checkmark & \checkmark \\
        \cline{2-10}
        & sjeng & & & \checkmark & \checkmark & \checkmark & \checkmark & & \checkmark \\
        \cline{2-10}
        & soplex & & & \checkmark & \checkmark & \checkmark & \checkmark & & \\
        \hline
        \hline
        \multirow{5}{*}{PTS} & bullet & & \checkmark & \checkmark & & \checkmark & \checkmark & & \checkmark \\
        \cline{2-10}
        & dcraw & & & \checkmark & & \checkmark & \checkmark & \checkmark & \checkmark \\
        \cline{2-10}
        & ffmpeg & \checkmark & \checkmark & \checkmark & & \checkmark & \checkmark & \checkmark & \checkmark \\
        \cline{2-10}
        & fhourstones & & \checkmark & \checkmark & & \checkmark & & \checkmark & \checkmark \\
        \cline{2-10}
        & vp8 & & & \checkmark & & \checkmark & \checkmark & & \checkmark \\
        \hline
    \end{tabular}
\end{footnotesize}
\end{small}
\vspace{-1.5mm}
\end{table} 


\mytable\ref{tab:OSR-alC-bench-desc} describes which LLVM transformations are effective on the hottest function from the benchmarks discussed in \mysection\ref{ss:evaluation}. CSE and LICM apply to all of them, and Sink to all but one benchmark (\mytt{fhourstones}). LCSSA-form construction is triggered by LICM in all benchmarks with the exception of \mytt{soplex}. 

\myfigure\ref{fig:CS-debug-tot-dead} presents results collected on the corpus of functions of the three largest benchmarks from our case study (\mysection\ref{ss:spec-benchmarks}). Our goal is to to investigate possible correlations between the size of a function and the number of user variables affected by source-level debugging issues. Each point in a scatter plot represents a function: the horizontal position is given by the number of IR instructions in its unoptimized code version, while the vertical position by the sum of the number of endangered user variables across program points corresponding to source-level locations.

\begin{figure}[!t]
\vspace{-1mm}
\begin{center}
\makebox[\textwidth][c]{
\includegraphics[width=0.54\textwidth]{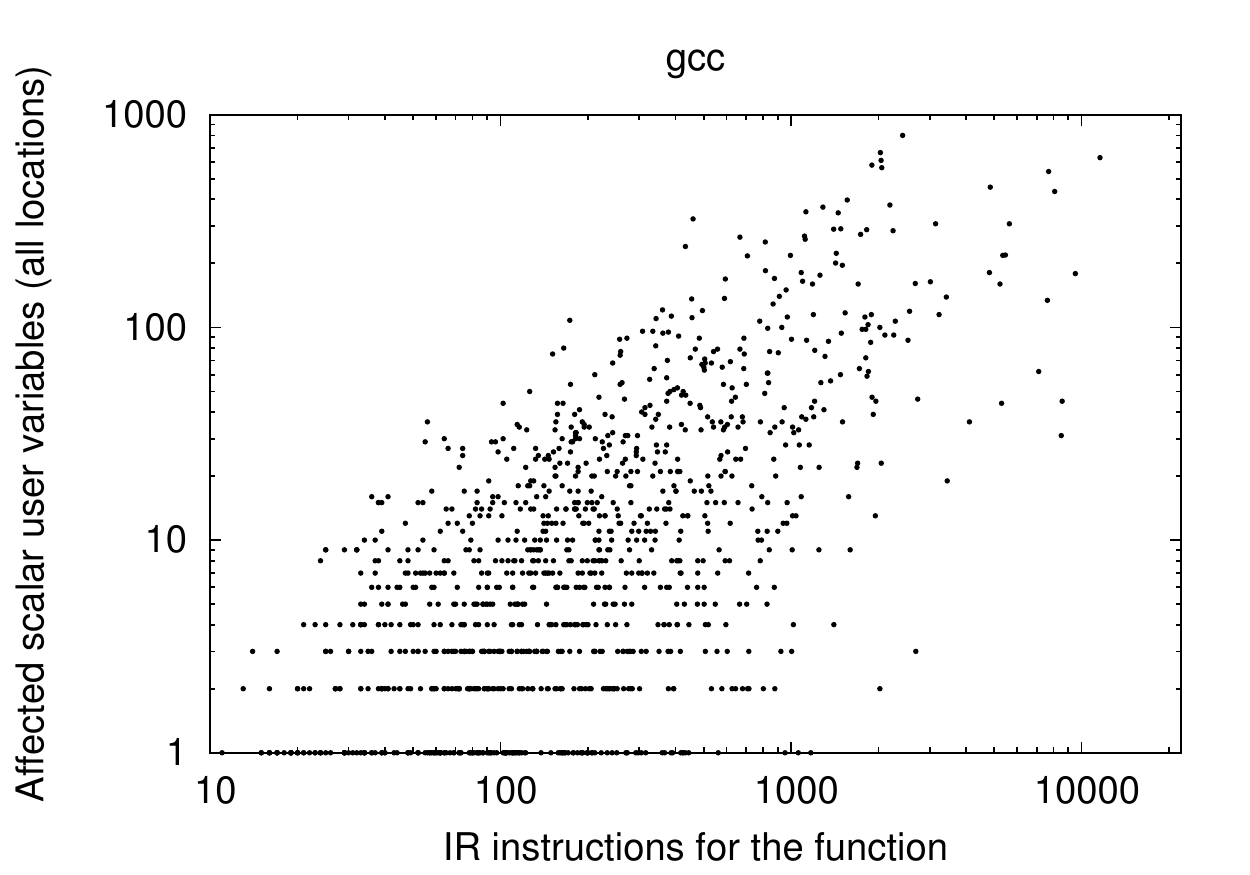}
\includegraphics[width=0.54\textwidth]{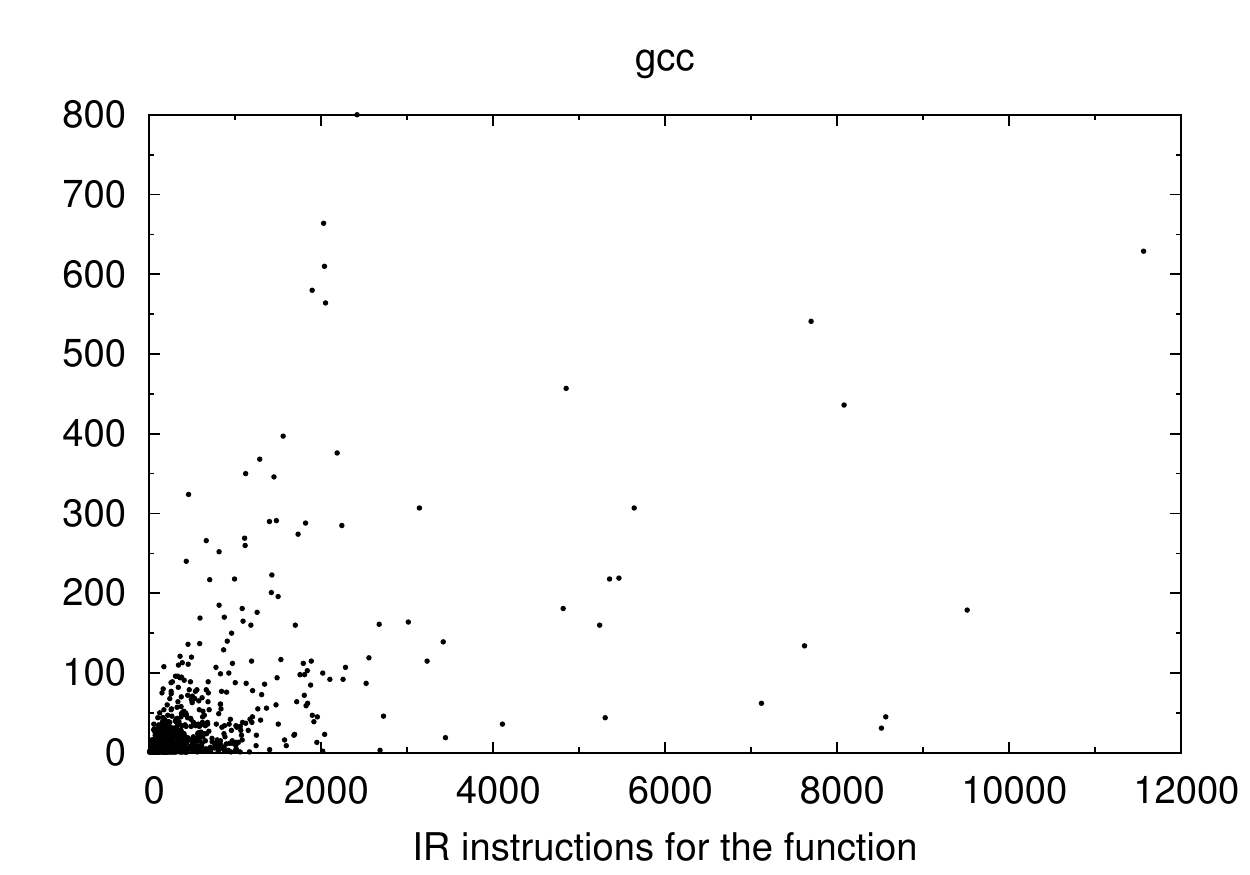}
}
\vspace{0.5mm}
\makebox[\textwidth][c]{
\includegraphics[width=0.54\textwidth]{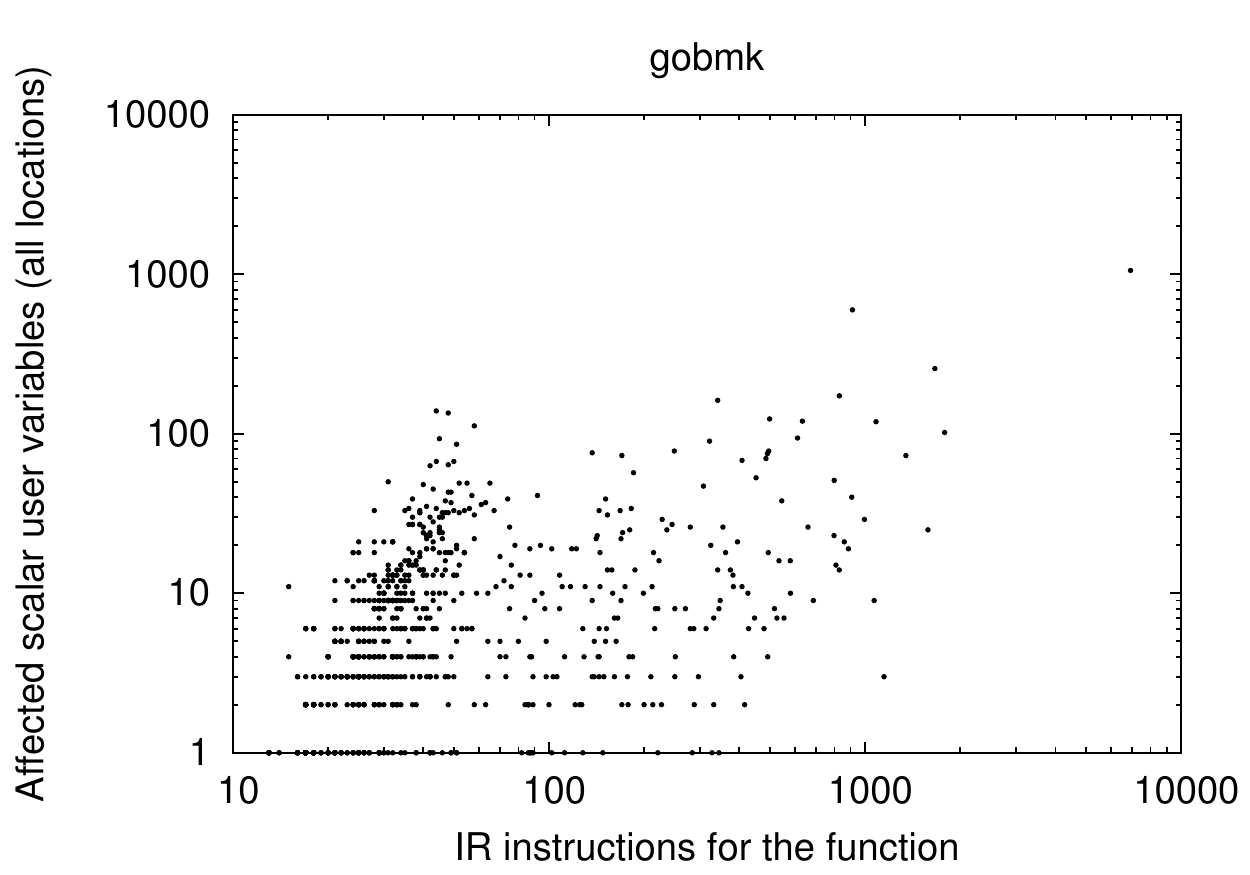}
\includegraphics[width=0.54\textwidth]{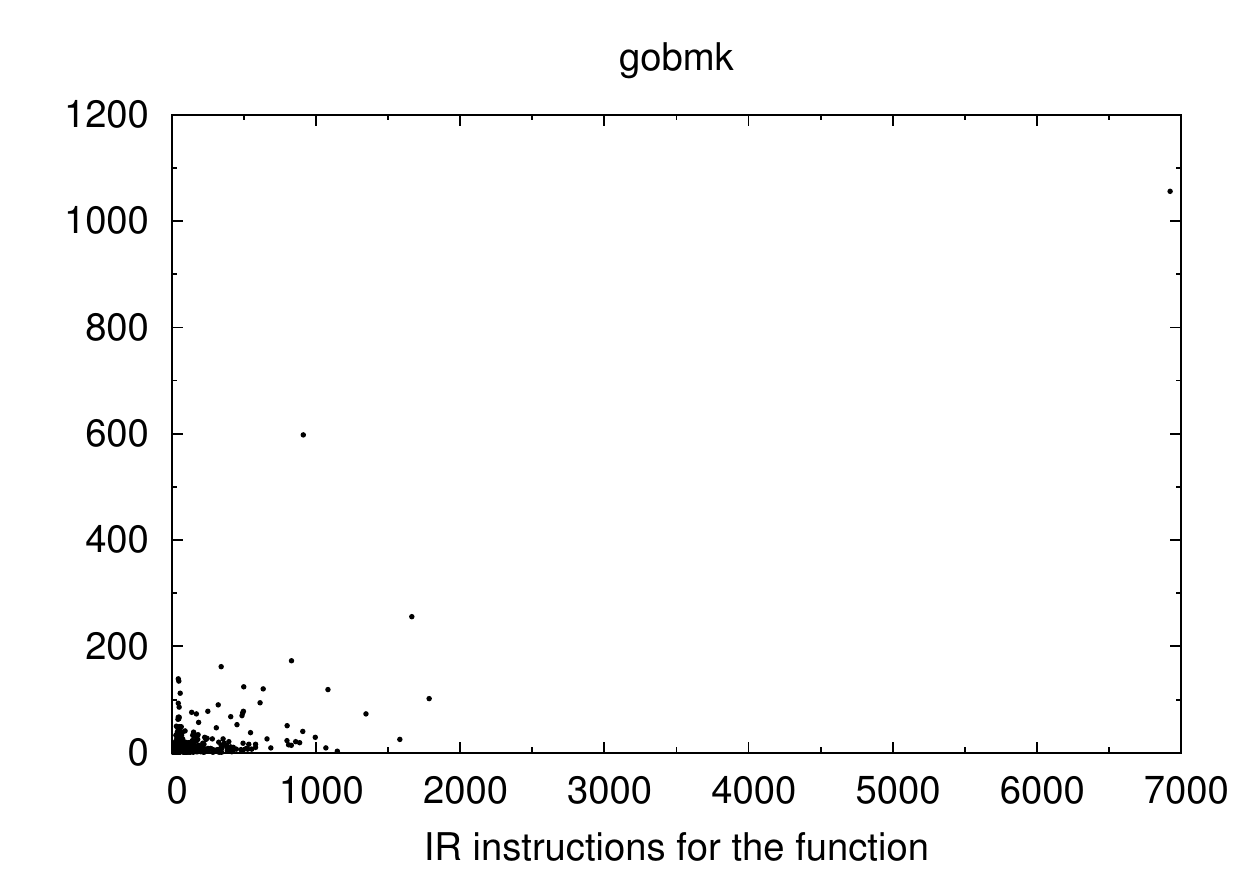}
}
\vspace{0.5mm}
\makebox[\textwidth][c]{
\includegraphics[width=0.54\textwidth]{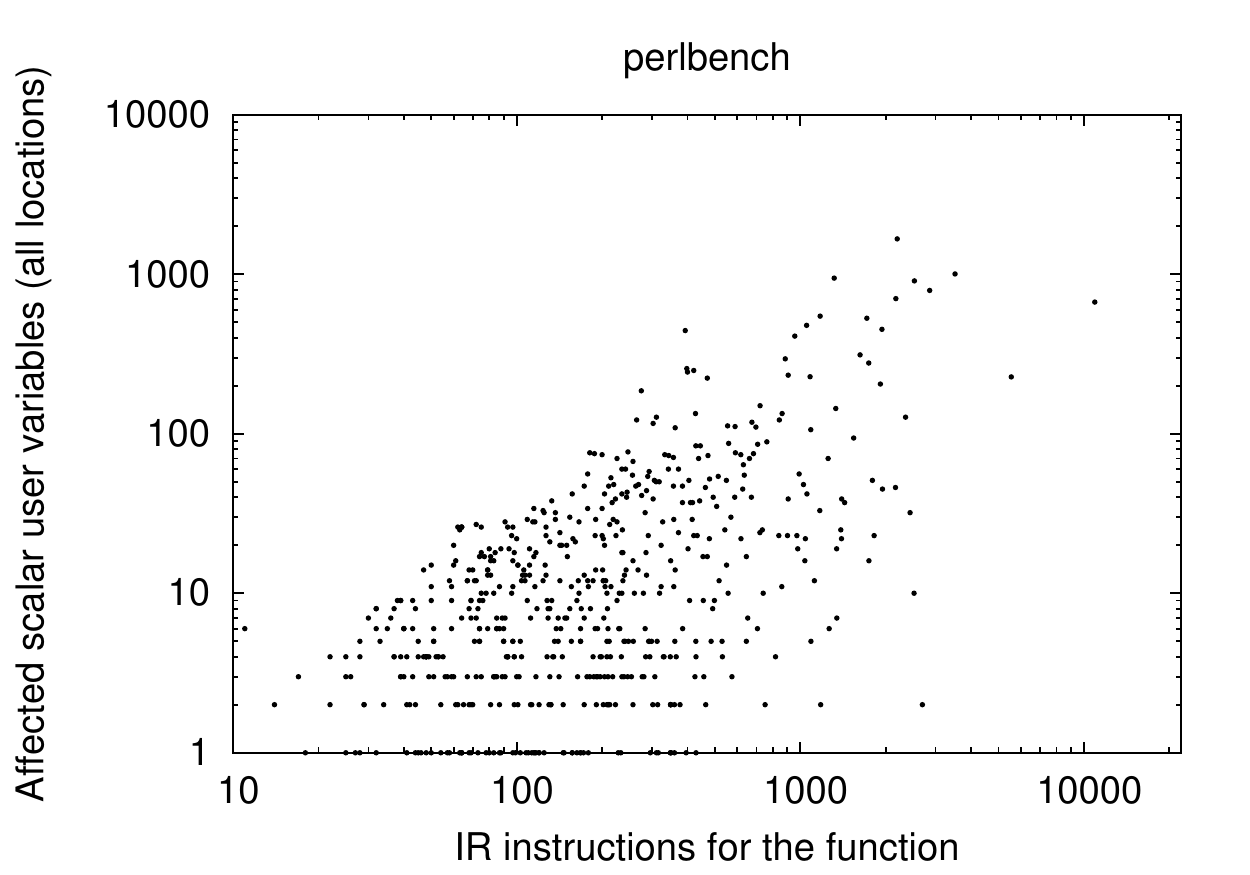}
\includegraphics[width=0.54\textwidth]{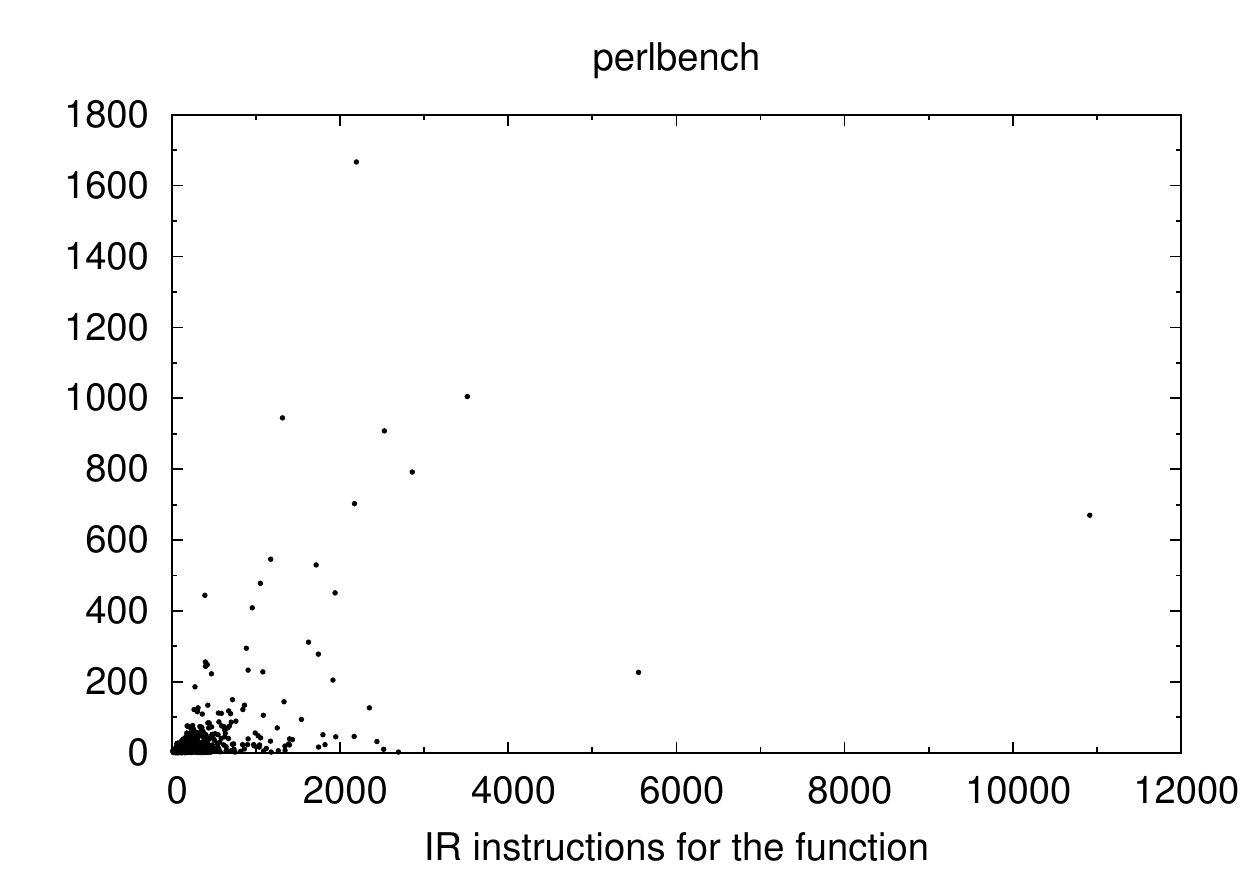}
}
\caption{\label{fig:CS-debug-tot-dead} Total number of endangered scalar user variables across program points. The horizontal coordinate of each function is determined by the size of its unoptimized version. For each benchmark we report a log-log (left) and a linear (right) plot.}
\end{center}
\vspace{-2mm}
\end{figure}

The log-log plots for \mytt{gcc} may suggest a trend line such that larger functions would typically have a large number of affected variables. However, this trend is less pronounced in \mytt{perlbench}, and nearly absent from \mytt{gobmk}. Linear plots should provide the reader with a better visualization of what happens for larger functions and for functions with a higher total number of affected variables. We can safely conclude that, although larger functions might be more prone to source-level debugging issues, these issues frequently arise for smaller functions as well.

\end{document}